\documentclass[11pt,a4paper]{article}
\usepackage{amsmath,amssymb,amsfonts,latexsym,graphicx,amsthm}
\usepackage{color}
\usepackage{url,hyperref}
\usepackage{comment}
\usepackage[linesnumbered,boxed,ruled,noend]{algorithm2e}
\usepackage{xcolor} \usepackage{hyperref} \hypersetup{
    colorlinks = true, allcolors = black, citecolor = {blue!0!black}, linkcolor = {blue!0!black}
}
\usepackage[shortlabels]{enumitem}
\usepackage{cleveref}

\usepackage[margin=1in]{geometry}

\newtheorem{theorem}{Theorem}[section]
\newtheorem{lemma}{Lemma}[section]
\newtheorem{definition}{Definition}[section]
\newtheorem{corollary}{Corollary}[section]
\newtheorem{claim}{Claim}[section]

\newtheorem{invariant}{Invariant}[section]

        {\medskip}

\newcommand{\eps}{\epsilon}

\newcommand{\ceil}[1]{\lceil #1 \rceil}
\newcommand{\floor}[1]{\lfloor #1 \rfloor}

\newcommand{\tree}{\mathcal{T}}
\newcommand{\freq}{\mathbf{f}}
\newcommand{\clr}{\mathcal{C}}
\newcommand{\cnt}{\mathsf{cnt}}

\newcommand{\brac}[1]{\left(#1\right)}
\newcommand{\field}{\mathbb{F}}
\newcommand{\mat}{\mathcal{C}}

\newcounter{paragraphCounter} 

\usepackage{tikz}
\usetikzlibrary{trees}
\usetikzlibrary{shapes.geometric}

\begin{document}

\title{Improved Streaming Edge Coloring}
\author{
	Shiri Chechik\thanks{Tel Aviv University, \href{}{shiri.chechik@gmail.com}} \and 
	Hongyi Chen\thanks{State Key Laboratory for Novel Software Technology, New Cornerstone Science Laboratory, Nanjing University, \href{}{hychener01@gmail.com}}\and 
	Tianyi Zhang\thanks{ETH Zürich, \href{}{tianyi.zhang@inf.ethz.ch}}
}
\date{}

\maketitle

\begin{abstract}
	Given a graph, an edge coloring assigns colors to edges so that no pairs of adjacent edges share the same color. We are interested in edge coloring algorithms under the W-streaming model. In this model, the algorithm does not have enough memory to hold the entire graph, so the edges of the input graph are read from a data stream one by one in an unknown order, and the algorithm needs to print a valid edge coloring in an output stream. The performance of the algorithm is measured by the amount of space and the number of different colors it uses.
	
	This streaming edge coloring problem has been studied by several works in recent years. When the input graph contains $n$ vertices and has maximum vertex degree $\Delta$, it is known that in the W-streaming model, an $O(\Delta^2)$-edge coloring can be computed deterministically with $\tilde{O}(n)$ space [Ansari, Saneian, and Zarrabi-Zadeh, 2022], or an $O(\Delta^{1.5})$-edge coloring can be computed by a $\tilde{O}(n)$-space randomized algorithm [Behnezhad, Saneian, 2024] [Chechik, Mukhtar, Zhang, 2024].
	
	In this paper, we achieve polynomial improvement over previous results. Specifically, we show how to improve the number of colors to $\tilde{O}(\Delta^{4/3+\epsilon})$ using space $\tilde{O}(n)$ deterministically, for any constant $\epsilon > 0$.
    This is the first deterministic result that bypasses the quadratic bound on the number of colors while using near-linear space.

\end{abstract}

\thispagestyle{empty}
\clearpage
\setcounter{page}{1}	

\section{Introduction}

Let $G = (V, E)$ be an undirected graph on $n$ vertices with maximum vertex degree $\Delta$. An edge coloring of $G$ is an assignment of colors to edges in $E$ such that no pairs of adjacent edges share the same color, and the basic objective is to understand the smallest possible number of colors that are needed in any edge coloring, which is called the edge-chromatic number of $G$. It is clear that the total number of colors should be at least $\Delta$, and a simple greedy algorithm can always find an edge coloring using $2\Delta-1$ colors. By the celebrated Vizing's theorem \cite{vizing1965chromatic} and Shannon's theorem \cite{shannon1949theorem}, $(\Delta+1)$-edge coloring and $\floor{3\Delta/2}$-edge coloring always exist in simple and multi-graphs respectively, and these two upper bounds are tight in some hard cases.

The edge coloring problem has been studied widely from the algorithmic perspective. There have been efficient algorithms for finding good edge coloring in various computational models, including sequential \cite{arjomandi1982efficient,gabow1985algorithms,sinnamon2019fast,Assadi24,BhattacharyaCCSZ24,BhattacharyaCSZ24,assadi2024vizing},  dynamic~\cite{BarenboimM17,BhattacharyaCHN18,duan2019dynamic,Christiansen23,BhattacharyaCPS24,Christiansen24}, online~\cite{CohenPW19,BhattacharyaGW21,SaberiW21,KulkarniLSST22,BilkstadSVW24,BlikstadOnline2025,dudeja2024randomizedgreedyonlineedge}, and distributed~\cite{panconesi2001some,elkin20142delta,fischer2017deterministic,ghaffari2018deterministic,balliu2022distributed,ChangHLPU20,Bernshteyn22,Christiansen23,Davies23} models. In this paper, we are particularly interested in computing edge coloring under the streaming model where we assume the input graph does not fit in the memory of the algorithm and can only be accessed via one pass over a stream of all edges in the graph. Since the output of edge coloring is as large as the graph size, the algorithm cannot store it in its memory. To address this limitation, the streaming model is augmented with an output stream in which the algorithm can write its answers during execution, and this augmented streaming model is thus called the W-streaming model.

Edge coloring in the W-streaming model was first studied in \cite{behnezhad2019streaming}, and improved by follow-up works \cite{charikar2021improved,ansari2022simple} which led to a deterministic edge coloring algorithm using $O(\Delta^2)$ colors and $O(n)$ space, or $O(\Delta^2/s)$ colors and $O(ns)$ space as a general trade-off. In \cite{ghosh2024low}, the authors improved the trade-off to $O(\Delta^2/s)$ colors and \footnote{$\tilde{O}(f)$ hides $\log f$ factors.}$\tilde{O}(n\sqrt{s})$ memory, yet it does not break the quadratic bound in the most natural $\tilde{O}(n)$ memory regime; this algorithm is randomized but can be derandomized in exponential time. The quadratic upper bound of colors was subsequently bypassed in \cite{behnezhad2023streaming,chechik_et_al:LIPIcs.ICALP.2024.40} where it was shown that an $O(\Delta^{1.5})$-edge coloring can be computed by a randomized algorithm using $\tilde{O}(n)$ space. Furthermore, in \cite{behnezhad2023streaming} the authors obtained a general trade-off of $O(\Delta^{1.5}/s + \Delta)$ colors and $\tilde{O}(ns)$ space which is an improvement over \cite{ghosh2024low}.

\subsection{Our Results}
In this paper, we focus on the basic $\tilde{O}(n)$-memory setting and improve the recent $\Delta^{1.5}$ randomized upper bound to $\Delta^{4/3+\epsilon}$.

\begin{theorem}\label{rand}
    Given a simple graph $G = (V, E)$ on $n$ vertices with maximum vertex degree $\Delta$, for any constant $\epsilon > 0$, there is a randomized W-streaming algorithm that outputs a proper edge coloring of $G$ using $O\brac{(\log\Delta)^{O(1/\epsilon)}n}$ space and $O\brac{(\log\Delta)^{O(1/\epsilon)}\Delta^{4/3+\epsilon}}$ different colors; both upper bounds hold in expectation.
\end{theorem}

Furthermore, we also show that our algorithm can be derandomized using bipartite expanders based on error correcting codes at the cost of slightly worse bounds, as stated below.

\begin{theorem}\label{main}
    Given a simple graph $G = (V, E)$ on $n$ vertices with maximum vertex degree $\Delta$, for any constant $\epsilon > 0$, there is a deterministic W-streaming algorithm that outputs a proper edge coloring of $G$ using $O\brac{(\log\Delta)^{O(1 / \epsilon)}\cdot(1/\epsilon)^{O(1/\epsilon^3)}\cdot \Delta^{4/3+\epsilon}}$ colors and $O\brac{n\cdot (\log \Delta)^{O(1 / \epsilon^4)}}$ space.
\end{theorem}

\subsection{Technical Overview}

\paragraph*{Previous Approaches.} Using a deterministic general-to-bipartite reduction from \cite{ghosh2024low}, we can assume the input graph $G = (L\cup R, E)$ is bipartite. Also, it suffices to color only a constant fraction of all edges in $G$, because we can recurse on the rest $1 - \Omega(1)$ fraction of $G$ which only incurs an extra factor of $O(\log\Delta)$ on the total number of different colors.

Let us begin by recapping the randomized $\tilde{O}(\Delta^{1.5})$-edge coloring from \cite{behnezhad2023streaming}. Since the algorithm has $\tilde{O}(n)$ bits of memory, we can assume that input graph is read from stream in batches of size $\Theta(n)$. If the subgraph formed by every batch has maximum degree at most $\Delta^{1/2}$, then we can allocate $O(\Delta^{1/2})$ new colors for each batch, using $O(\Delta^{1.5})$ colors in total. Thus, the main challenge arises when some batches contain subgraphs with maximum degree exceeding $\Delta^{1/2}$.

To simplify the problem, let us assume that in each batch of $O(n)$ edges, every vertex in $R$ has degree $d > \Delta^{1/2}$. To assign colors to edges, organize a table of colors of size $\Delta\times (\Delta/d)$, represented as a matrix $C[i, j]$ with indices $ 1\leq i\leq \Delta$ and $1\leq j\leq \Delta /d$. Then, for every vertex $u\in L$, draw a random shift $r_u$ uniformly at random from $[\Delta]$. During the algorithm, each vertex $u\in L$ keeps a counter $c_u$ of its degree in the stream so far, and each vertex $v\in R$ keeps a counter $b_v$ of the number of batches in which it has appeared so far. Then, to color a single batch, for edge $(u, v)$, the algorithm tentatively assigns the color $C[r_u + c_u, b_v]$ (indices are under modulo $\Delta$).

Clearly, edges incident on the same vertex $u\in L$ receive different colors, because the counter $c_u$ is incremented for each edge $(u, v)$; also edges incident on the same vertex $v\in R$ but arriving in different batches are different, because the values of counter $b_v$ are different. Randomization ensures that edges incident on $v\in R$ within the same batch receive mostly  different colors from the same column in $C[*, b_v]$. Consider all the $d$ neighbors of $v$ in a single batch $u_1, u_2, \ldots, u_d$. Since all the random shifts $r_{u_i}$ are independent and all the counters $c_{u_i}$ are deterministic (they only depend on the input stream), with high probability, most row indices $(r_{u_i} + c_{u_i}) \bmod \Delta$ will be different. 

\paragraph*{Bypassing the $\Delta^{1.5}$ Bound.} To better understand the bottleneck in the approach of \cite{behnezhad2023streaming}, consider the following case. If every batch forms a regular subgraph with uniform degree $d$, then we can reduce the size of the color table from $\Delta\times (\Delta/d)$ to $(\Delta/d)\times (\Delta/d)$, since each vertex $u\in L$ could only appear in $\Delta/d$ different batches. So the main difficult case is when the batches are subgraphs of \textbf{unbalanced degrees}. As an extreme example, consider when vertices in $L$ have degree $1$, while the vertices in $R$ have degree $d$. For the rest, our main focus will be on this extreme case, and show how to obtain a $\Delta^{1+\epsilon}$-edge coloring which is almost optimal.

The flavor of our approach is similar to \cite{chechik_et_al:LIPIcs.ICALP.2024.40}. Divide all $\Delta$ batches into $\Delta / d$ phases, each phase consisting of $d$ consecutive batches. Let $D$ be a parameter which upper bounds the number of batches in which any vertex $v\in R$ could appear during a single phase. Then, the maximum degree of a single phase would be bounded by $Dd$, so in principle we should be able to color all edges within this phase using $O(Dd)$ different colors. To implement the coloring procedure, at the beginning of each phase, prepare $D$ fresh palettes $C_1, C_2, \ldots, C_{D}$, each of size $d$, and assign each batch in this phase a palette $C_i$ where $i$ is chosen from $[D]$ uniformly at random. To keep track of the colors already used around each vertex, we maintain the following data structures.
\begin{itemize}
	\item Each vertex $v\in R$ keeps a list $U_v\subseteq \{C_1, C_2, \ldots, C_D\}$ of used palettes.
	\item Each vertex $u\in L$ initializes a random shift $r_u\in [d]$ at the beginning of the algorithm.
\end{itemize} 

When a batch of edges $F\subseteq E$ arrives, we will use the assigned palette $C_i$ to color this batch. For every edge $(u, v)\in F$, if $C_i\notin U_v$, then tentatively assign the color $(r_u + c_i) \bmod d$ in $C_i$ to edge $(u, v)$, where $c_i$ denotes the number of times that palette $C_i$ has been assigned to previous batches. If $C_i \in U_v$, we mark the edge as uncolored. Since the palettes are assigned to batches randomly, in expectation, a constant fraction of edges will be successfully colored. In the case that multiple edges incident on $v$ are assigned the same color, we retain only one such edge and mark the remaining edges as uncolored. Since all the random shifts $r_u$'s are uniformly random and independent of the randomness of counters $c_i$, most tentative colors around any vertex $v\in R$ in this batch should be different. Once the batch $F$ is processed, add $C_i$ to all lists $U_v, v\in R$ such that $\deg_F(v) = d$.

To reason about the space usage, we can argue that the total size of the lists $U_v, \forall v\in R$ does not exceed $O(n)$, because a palette $C_i$ appears in the list $U_v$ only when $v$ receives $d$ edges in a batch which uses palette $C_i$. Since the total number of edges in a phase is $O(dn)$, the total list size should be bounded by $O(n)$. In this way, we are able to store all the used palettes of vertices in $R$; this is not new to our approach, and a similar argument was also used in \cite{chechik_et_al:LIPIcs.ICALP.2024.40}. The new ingredient is the way we store the used colors around vertices in $L$. Here we have exploited the fact that vertices in $L$ have low degrees in each batch, so their progress in every palette $C_i$ is \textbf{synchronized}; that is, we only need to store a single counter $c_i$ which is the number of times that $C_i$ appears, and then the next tentative color of $u\in L$ would be $(r_u + c_i) \bmod d$. 

The above scheme would use $O(Dd)$ different colors in a single phase, so $O(\Delta D)$ colors across all $\Delta /d$ phases. When $D\leq \Delta^{1/4}$, this bound would be much better than $\Delta^{1.5}$. So, what if $D > \Delta^{1/4}$? To deal with this case, we will apply a two-layer approach. Specifically, let us further group every $D$ consecutive phases as a super-phase, so there are at most $\Delta / Dd < \Delta^{1/4}$ super-phases in total (recall that we were assuming $d > \Delta^{1/2}$). Within a super-phase, we will allocate $\Delta$ fresh colors which are divided into $\Delta / Dd$ different color packages $P_1, P_2, \ldots, P_{\Delta / Dd}$, where each color package $P_i$ is further divided into $D$ palettes of size $d$ as $P_i = C_{i, 1}\cup C_{i, 2}\cup\cdots\cup C_{i, D}$. In this way, the total number of colors would be $\Delta^{5/4}$. 

Then, like what we did before, for each phase in a super-phase, assign a color package $P_i$ from $P_1, P_2, \ldots, P_{\Delta / Dd}$ uniformly at random. Within each phase, we will stick to its assigned color package $P_i$ and reuse the algorithm within a phase we have described before. Since a color package is shared by multiple phases, each vertex $v\in R$ needs to store a list $U_{v, 2}$ which stores all the packages $P_i$ it has used so far, and in any phase where $P_i$ is assigned but $P_i$ is already contained in $U_{v, 2}$, we would not color any edges incident on $v$ in this particular phase. By repeating the same argument, we can argue that the total size of all the lists $U_{v, 2}, \forall v\in R$ is bounded by $O(n)$ as well.

We can further generalize this two-layer approach to multiple layers and reduce the total number of colors to $\Delta^{1 + \epsilon}$. However, this only works with the most unbalanced setting where vertices on $L$ always have constant degrees in each batch of input. In general, when low-degree side has degree $d'$, our algorithm needs $d'\cdot\Delta^{1+\epsilon}$ colors. If $d'$ is large, then we would switch to the color table approach from \cite{behnezhad2023streaming}. Balancing the two cases, we end up with $\Delta^{4/3+\epsilon}$ colors overall. 

\paragraph*{Derandomization using Bipartite Expanders.} Randomization is used both in the unbalanced case and in the regular case. To replace the choices of the random shifts $r_u$ and random color package assignments, we will show one can apply unbalanced bipartite expanders \cite{ta2001loss,guruswami2009unbalanced,kalev2022unbalanced} in a black box manner. However, for the random shifts used in the regular case where we apply the color table idea from \cite{behnezhad2023streaming}, it would not be enough to use an arbitrary bipartite expander, because the counters $c_{u}$'s could be different and possibly damage the expansion guarantee; for example, it is not clear to us how to apply the bipartite expander construction based on Parvaresh–Vardy Codes \cite{guruswami2009unbalanced}. To fix this issue, it turns out that it would be most convenient to use the bipartite expander construction based on multiplicity codes \cite{kalev2022unbalanced}.

\section{Preliminaries}
\paragraph*{Basic Terminologies.} For any integer $k$, let us conventionally define $[k] = \{1, 2, \ldots, k\}$. For any set $S$ and integer $k$, let $k\odot S$ be the multi-set that contains exactly $k$ copies of each element in $S$.

Let $G = (V, E)$ denote the simple input graph on $n$ vertices and $m$ edges with maximum degree $\Delta$. For any subset of edges $F\subseteq E$ and any vertex $u\in V$, let $\deg_F(u)$ be the number of edges in $F$ incident on $u$. Sometimes we will also refer to the subgraph $(V, F)$ simply as $F$.

\paragraph*{Problem Definition.} In the edge coloring problem, we need to find an assignment of colors to edges such that adjacent edges have distinct colors, and the objective is to minimize the total number of different colors.

In the W-streaming model introduced by \cite{demetrescu2009trading,glazik2017finding}, all edges of the graph $G$ arrive one by one in the stream in an arbitrary order, and the algorithm makes one pass over the stream to perform its computation. For the task of edge coloring, since the algorithm has much less space than the total size of the graph, it cannot store all the edge colors in its memory. To output the answer, the algorithm is given a write stream in which it can print all colors.

Next, to set the stage in a convenient way, we will state several reductions for the problem.

\paragraph*{General-to-Bipartite Reduction.} Let us first simplify the problem by a deterministic reduction from edge coloring in general graphs to bipartite graphs.
\begin{lemma}[Corollary 3.2 in \cite{ghosh2024low}]
	Given an algorithm $\mathcal{A}$ for streaming edge coloring on bipartite graphs using $f(\Delta)$ colors and $g(n, \Delta)$ space, there is an algorithm $\mathcal{B}$ using $O(f(\Delta))$ colors and $O\brac{g(n, \Delta)\log n + n\log n\log\Delta}$ space in general graphs. Furthermore, this reduction is deterministic.
\end{lemma}

\paragraph*{Reduction to Fixed Degree Pairs.} By the above reduction, we focus on edge coloring for bipartite graphs $G = (V, E)$, where $V = L \cup R$. Since the algorithm has space $\Omega(n)$, we can divide the input stream into at most $ O(m / n) = O(\Delta) $ batches, each batch containing $n$ edges. For any vertex $u\in V$ and any batch $F$, let $\deg_F(u)$ be the number of edges in $F$ incident on $u$. For every pair of integers $0\leq l, r\leq \log_2\Delta$ and every batch $F$, let $F_{l, r} = \{(u, v)\in F\mid \deg_F(u)\in [2^l, 2^{l+1}), \deg_F(v) \in [2^r, 2^{r+1})\}$, and define $E_{l, r} = \bigcup_{F}F_{l, r}$ over all batches $F$ in the stream and $m_{l, r} = |E_{l, r}|$.

In the main text, we will devise an algorithm to handle edges in $F_{l, r}$ for any fixed pair of $(l, r)$. The final coloring is obtained by taking the union of the colors over all $(l, r)$, which will blow up the number of colors and space by a factor of $O(\log^2\Delta)$.

\paragraph*{Adapting to an Unknown $\Delta$.} So far we have assumed that the maximum vertex degree $\Delta$ in $G$ is known in advance. Our algorithm can be adapted to an unknown value $\Delta$ in a standard way. Specifically, we can maintain the value $\Delta_t$ which is the maximum degree of the subgraph containing the first $t$ edges in the input stream. Whenever $\Delta_t\in (2^{k-1}, 2^{k}]$, we apply our algorithm with $\Delta = 2^k$ to color all the edges. When $k$ increments at some point, we restart a new instance of the algorithm with a new choice $\Delta = 2^k$ and continue to color the edges with fresh colors. Overall, the total number of colors will not change asymptotically.
\section{Randomized $\Delta^{4/3+\epsilon}$ Edge Coloring}

This section is devoted to the proof of \Cref{rand}.

\paragraph*{Reduction to Partial Coloring.} To find an edge coloring of a graph, it was shown in \cite{chechik_et_al:LIPIcs.ICALP.2024.40} that it suffices to color a constant fraction of edges in $E$ if we do not care about $O(\log\Delta)$ blow-ups in the number of colors. Roughly speaking, for the uncolored edges, we can view them as another instance of edge coloring, and solve it recursively. The following statement formalizes this reduction.
\begin{lemma}[implicit in \cite{chechik_et_al:LIPIcs.ICALP.2024.40}]
	Suppose there is a randomized streaming algorithm $\mathcal{A}$ with space $g(n, \Delta)$ space, such that for each edge in $e\in E$, it assigns a color from $[f(\Delta)]$ or marks it as $\bot$ (uncolored) and print this information in the output stream, with the guarantee that there are at least $\delta m$ edges in expectation which receive actual colors, for some value $\delta > 0$. Then, there is a randomized edge coloring algorithm $\mathcal{B}$ which uses at most $O\brac{\frac{\log\Delta}{\delta}f(\Delta)}$ colors and $O\brac{\frac{\log\Delta}{\delta}g(n, \Delta)}$ space in expectation.
\end{lemma}
\begin{proof}[Proof sketch]
    The idea is to recursively apply the streaming algorithm on all edges marked with $\bot$ (uncolored). In expectation, each time we reapply the algorithm, the number of uncolored edges decrease by a factor of $1-\delta$. So the expected recursion depth would be $O(\frac{\log\Delta}{\delta})$ before all uncolored edges fit in memory.
\end{proof}


According to the reductions in the preliminaries, we will focus on partial edge coloring in bipartite graphs with fixed degree pairs. More specifically, our algorithm consists of the following two components.

\begin{lemma}\label{low}
	Fix a parameter $\epsilon > 0$. Given an graph $G = (V, E)$ on $n$ vertices with maximum vertex degree $\Delta$, for any constant $\epsilon > 0$, and fix an integer pair $(l, r)$, there is a randomized W-streaming algorithm that outputs a partial coloring of edges $F_{l, r}$ such that least $\delta m_{l, r}$ edges receive colors in expectation; here $\delta = 2^{-O(1/\epsilon)}$ is also a constant. The algorithm uses $O\brac{(\log\Delta)^{O(1 / \epsilon)}\cdot \Delta^{1+\epsilon}\cdot 2^l}$ colors and $O\brac{(\log\Delta)^{O(1 / \epsilon)}\cdot n}$ space.
\end{lemma}

\begin{lemma}\label{high}
	Given an graph $G = (V, E)$ on $n$ vertices with maximum vertex degree $\Delta$, and fix an integer pair $(l, r)$, there is a randomized W-streaming algorithm that outputs a partial coloring of edges $F_{l, r}$ such that least $m_{l, r}/2$ edges receive colors in expectation. The algorithm uses $O\brac{\Delta + \Delta^2 / 2^{l+r}}$ colors and $O(n)$ space.
	
\end{lemma}

\begin{proof}[Proof of \Cref{rand}]
	Basically, \Cref{low} deals with the case when $\min\{2^l, 2^r\}$ is small, and \Cref{high} deals with the case when $\min\{2^l, 2^r\}$ is large. For any $(l, r)$, if $\min\{2^l, 2^r\}\leq \Delta^{1/3}$, then by \Cref{low}, the number of colors is at most $O\brac{(\log\Delta)^{O(1 / \epsilon)}\cdot \Delta^{4/3+\epsilon}}$, and otherwise by \Cref{high} the number of colors would be $O\brac{\Delta^{4/3}}$. Either way, the total number of colors over all $(l, r)$ is bounded by $O\brac{(\log\Delta)^{O(1/\epsilon)}\Delta^{4/3+\epsilon}}$.
\end{proof}

\subsection{Proof of \Cref{low}}
\subsubsection{Data Structures}\label{Btree}

Before presenting the details of the data structures we will use in the main algorithm, let us start with a brief technical overview. The data structures consist of three components.

\begin{itemize}
	\item \textbf{Forest structures on batches.} This part organizes edge input batches into parameterized forests in a way similar to B-trees. The height of the forests will be $O(1/\eps)$, and the branch size of each level will be an integer power of $2$ in the range $[\Delta^\eps, \Delta]$, and so the total number of different forests will be bounded by $(\log\Delta)^{O(1/\eps)}$. 
	
	\item \textbf{Color allocation on forests.} Each forest will be associated with a separate color set of size roughly $\Delta^{1+\eps}$. On each level of a forest, we will randomly partition the color set of this forest into color subsets (packages) and assign them to the tree nodes on this level. We will also make sure that the color packages on tree nodes are nested; that is, the color package of a node is a subset of the color package of its parent node.
	
	Each vertex will choose colors for its incident edges according to its own frequencies of accumulating edges in the stream. For example, when a vertex is gathering a large number of incident edges in a short interval of batches, it would use color packages on tree nodes with large branch sizes.
	
	\item \textbf{Vertex-wise data structures.} To ensure that we never assign the same color to adjacent edges, each vertex needs to remember which colors it has already used around it. To efficiently store all the previously used colors, we will show that the used colors are actually concentrated in color packages, so each vertex only needs to store the tree nodes corresponding to those color packages, instead of storing every used colors individually.
\end{itemize}

Next, let us turn to the details of the data structures we have outlined above.

\paragraph*{Forest Structures on Batches.} Without loss of generality, assume $l\leq r$, $2^r > \Delta^\epsilon$, and $\Delta^\epsilon$ is an integer power of two; note that if $2^r \leq \Delta^\epsilon$, then the maximum degree inside each batch $F_{l, r}$ is at most $\Delta^\epsilon$, so we can use a fresh palette of size $O(\Delta^\epsilon)$, so the total number of colors would be $O(\Delta^{1+\epsilon})$.

As we have done in the preliminaries, partition the entire input stream into at most $m / n\leq \Delta$ batches of size $n$. We will create at most $(\log\Delta)^{O(1/\epsilon)}$ different forest structures, where each leaf represents a batch, and the internal tree nodes represent consecutive batches. Each forest is parameterized by an integer vector $\freq = (f_1, f_2, \ldots, f_{h})$ such that:
\begin{itemize}
	\item $f_i$ is an integer power of two;
	\item $2^{r+1}=f_0\geq f_1\geq f_2\geq \cdots \geq f_h = \Delta^\epsilon$, and $2^{r+1}\cdot \prod_{i=1}^{h-1} f_i\leq m/n$.
\end{itemize}

We can assume the total number of batches in the input stream is an integer multiple of $2^{r+1}\cdot \prod_{i=1}^{h-1} f_i\leq m/n$ by padding empty batches. Given such a vector $\freq = (f_1, f_2, \ldots, f_{h})$, we will define a forest structure $\tree_\freq$ with $h+1$ levels by a bottom-up procedure; basically we will build a forest on all the batches with branching factors $2^{r+1}, f_1, \ldots, f_h$ bottom-up. More specifically, consider the following inductive procedure.
\begin{itemize}
	\item All the batches will be leaf nodes on level $0$. Partition the sequence of all batches into groups of consecutive $f_0 = 2^{r+1}$ batches. For each group, create a tree node at level-$1$ connecting to all leaf nodes in that group.
	
	\item Given any $1\leq i\leq h-1$, assume we have defined all the tree nodes on levels $1\leq j\leq i$. List all the tree nodes on level $i$ following the same ordering of the batches, and partition this sequence of level-$i$ nodes into groups of size $f_i$. For each group, create a tree node at level-$(i+1)$ that connects to all nodes in the group.
\end{itemize}
According to the definition, in general, tree levels up to $k$ have the same topological structure for all frequency vectors which share the same first $k-1$ coordinates $f_1, f_2, \ldots, f_{k-1}$. For any node $N\in V(\tree_\freq)$, let $\tree_\freq(N)$ be the subtree rooted at node $N$. By the above definition, for any $1\leq k\leq h$, for any level-$k$ node $N$, the set of all leaf nodes contained in the subtree $\tree_\freq(N)$ form a sub-interval of the batch sequence with length $2^{r+1}\cdot \prod_{i=1}^{k-1}f_i$. 

\paragraph*{Color Allocation on Forests.} Next, we will allocate color packages at each tree node of each forest $\tree_\freq$. By construction, each forest structure $\tree_\freq$ is a tree of $h+1$ levels (from level-$0$ to level-$h$), and each tree is rooted at a level-$h$ node. Go over each tree $T\subseteq \tree_\freq$ and we will allocate color packages to tree nodes in a top-down manner.
\begin{itemize}
	\item \textbf{Basis.} At the root $N$ of the tree $T$, allocate a color package $\clr^N$ with $25\cdot 2^{l+r+2}\cdot \prod_{i=1}^h (5f_i)$ new colors. Divide package $\clr^N$ evenly into $5f_h =5\Delta^\epsilon$ smaller packages (colors are ordered alphabetically in a package):
	$$\clr^N= \clr^N_1 \sqcup \clr^N_2 \sqcup \cdots \sqcup \clr^N_{5\Delta^\epsilon}$$
	Here, symbol $\sqcup$ means disjoint union. By construction of tree $T$, $N$ has $f_{h-1}$ different children $N_1, N_2, \ldots, N_{f_{h-1}}$. Let sequence $(i_1, i_2, \ldots, i_{5f_{h-1}})$ be a random permutation of $(f_{h-1} / \Delta^\epsilon)\odot [5\Delta^\epsilon]$. For each $1\leq j \leq f_{h-1}$, define color package $\clr^{N_j} = \clr^{N}_{i_j}$.
	
	\item \textbf{Induction.}
    In general, suppose we have defined color packages for tree nodes on levels $k, k+1, \ldots, h$. Go over all tree nodes on level $k$. Inductively, assume $|\clr^N| = 25\cdot 2^{l+r+2}\cdot \prod_{i=1}^k (5f_i)$. Divide color package $\clr^N$ into $5f_k$ smaller packages (colors are ordered alphabetically in a package):
	$$\clr^N = \clr^N_1\sqcup \clr^N_2\sqcup \cdots \sqcup \clr^N_{5f_k}$$
	Let $i_1, i_2, \ldots, i_{5f_{k-1}}$ be a random permutation of $(f_{k-1} / f_k)\odot [5f_k]$. For each such level-$k$ node $N$, by construction, it has $f_{k-1}$ children $N_1, N_2, \ldots, N_{f_{k-1}}$. 
	Then, for each $1\leq j\leq f_{k-1}$, define color package $\clr^{N_j} = \clr^N_{i_j}$.
\end{itemize}
By the above construction, each leaf node is allocated a color package of size $25\cdot 2^{l+r+2}$. We will call each such smallest color package a \textbf{palette}. Notice that, by definition, the same palette may appear at multiple leaf nodes (which represent input batches). For any leaf node $N$ (or equivalently, a batch), let $\cnt(\clr^N)$ count the total number of times that palette $\clr^N$ is also allocated to previous leaf nodes (batches). Note that the counters $\cnt(*)$ do not depend on the input stream and can be computed in advance.

\paragraph*{Vertex-Wise Data Structures.} To carry out the streaming algorithm, we will also maintain some data structures for vertices in $V$. At the beginning of the streaming algorithm, for each vertex $u\in L$, draw a random shift $r_u\in [3\cdot 2^{r+1}]$ uniformly at random; these random shifts $r_u$'s will remain fixed throughout the entire execution of the algorithm.

The main part is to specify the data structures associated with each vertex $v\in R$. For each forest $\tree_\freq$, we will maintain a set of marked nodes $M_{v, \freq}\subseteq V(\tree_\freq)$ throughout the streaming algorithm.
\begin{invariant}\label{inv}
	We will ensure the following properties regarding the marked nodes $M_{v, \freq}$ throughout the execution of the streaming algorithm.
	\begin{enumerate}[(1)]
		\item No two nodes in $M_{v, \freq}$ lie on the same root-to-leaf path in the forest $\tree_{\freq}$. Furthermore, suppose the current input batch corresponds to a leaf $F$, and let $P$ be the unique tree path from $F$ to the tree root. Then, any node $N\in M_{v, \freq}$ is a child of some node on $P$.
		
		\item For any node $N\in \tree_\freq$ on level-$k$ such that $M_{v, \freq}\cap V(\tree_\freq(N))\neq \emptyset$ , let $F_1, F_2, \ldots, F_s\subseteq E$ be all the input batches which correspond to leaf nodes in subtree $\tree_\freq(N)$. Take the union of the batches $U = F_1\cup F_2\cup\cdots \cup F_s$. Then, we have $\deg_U(v)\geq 2^{r-k}\cdot\prod_{i = 1}^k f_i$.
		
		\item For any previous input batch $F'$ before $F$ such that:
		\begin{itemize}
			\item $F$ and $F'$ are in the same tree component in $\tree_\freq$,
			\item $\deg_{F'}(v)\in [2^r, 2^{r+1})$,
			\item $v$ used some colors in $\clr^{F'}$ during the algorithm,
		\end{itemize} 
		we guarantee that $F'$ must be contained in some subtree $\tree_{\freq}(N)$ for some $N\in M_{v, \freq}$.
		
		\item The choices of $M_{v, \freq}$ is independent of the randomness of $\{r_u\mid u\in L\}$ and the randomness of color packages $\clr^*$, and they only depend deterministically on the input stream.
	\end{enumerate}
\end{invariant}

Let us explain the purpose of the above properties. \Cref{inv}(1)(2) ensures that the algorithm only uses $\tilde{O}(n)$ space in total. \Cref{inv}(3) allows us to rule out all colors used previously, preventing duplicate assignments to edges incident on the same vertex. \Cref{inv}(4) will be technically important for the analysis of randomization. 

\subsubsection{Algorithm Description}\label{rand-alg}
Next, let us turn to describe the coloring procedure. At the beginning, all the marked sets $M_{v, \freq}$ are empty for any $v\in R$ and any frequency vector $\freq$. Upon the arrival of a new input batch $F$, we will describe the procedures that update the marked sets and assign edge colors.

\paragraph*{Preprocessing Marked Sets.} Since the current input batch has changed due to the arrival of $F$, we may have violated \Cref{inv}(2) as the root-to-leaf tree path may have changed. Therefore, we first need to update all the marked sets with the following procedure named $\textsc{UpdateMarkSet}(F)$\label{alg:UpdateMarkSet}.

Go over every vertex $v\in R$ and every frequency vector $\freq$. Consider the position of $F$ in the forest $\tree_\freq$, and let $P$ be the root-to-leaf path in $\tree_{\freq}$ ending at leaf $F$. First, remove all marked nodes $N\in M_{v, \freq}$ which are not in the same tree as $P$; note that this may happen when $F$ is the first leaf in a new tree component of $\tree_{\freq}$.

Next, go over every node $W$ lying on the tree path $P$, for any child node $N$ of $W$, if (1) $V(\tree_\freq(N)) \not\ni F$ and (2) $V(\tree_\freq(N))\cap M_{v, \freq} \neq \emptyset$, then remove all nodes in $V(\tree_\freq(N))\cap M_{v, \freq}$ from $M_{v, \freq}$ and add $N$ to $M_{v, \freq}$. In other words, we elevate the positions of all the marked nodes in $M_{v, \freq}$ to their ancestors which are children of $P$. See \Cref{fig:pre} for an illustration.


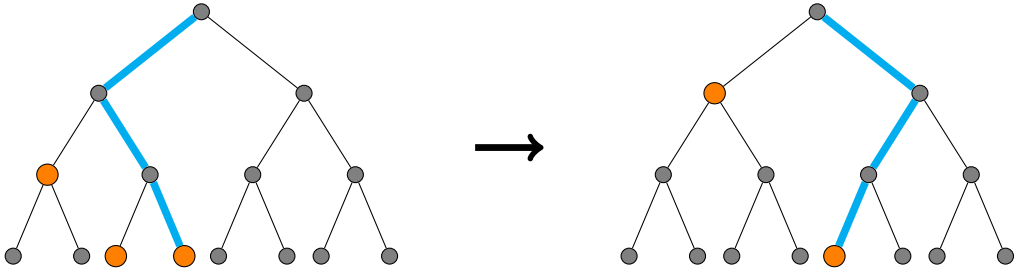
\begin{figure}[h!]
	\centering
    \begin{tikzpicture}[
  level distance=1.2cm,
  level 1/.style={sibling distance=3cm},
  level 2/.style={sibling distance=1.5cm},
  level 3/.style={sibling distance=1cm},
  every node/.style={circle, draw, fill=black!50,
	inner sep=0pt, minimum width=6pt},
  marked/.style={circle, draw=black, fill=orange, inner sep=1pt, minimum size=8pt},    
  triangle/.style={shape=isosceles triangle, shape border rotate=90,
                   minimum height=1cm, draw=gray!70, fill=gray!10, inner sep=0pt},
    scale = 0.9
  ]

\node(u0) {}
  child {node(u1) {}
    child {node[marked] {}
      child {node {}}
      child {node {}}
    }
    child {node(u2) {}
      child {node[marked] {}}
      child {node[marked] (u3) {}}
    }
  }
  child {node {}
    child {node {}
      child {node {}}
      child {node {}}
    }
    child {node {}
      child {node {}}
      child {node {}}
    }
  };

  \draw[->, thick, line width = 0.9mm] (4, -2) to (5, -2);

\draw[cyan, line width=1mm] (u0) -- (u1);  
\draw[cyan, line width=1mm] (u1) -- (u2);
\draw[cyan, line width=1mm] (u2) -- (u3);

\node(u0b) at (9, 0) {}  
  child {node[marked]{}
    child {node {}
      child {node {}}
      child {node {}}
    }
    child {node {}
      child {node {}}
      child {node {}}
    }
  }
  child {node(u1b) {}
    child {node(u2b) {}
      child {node[marked](u3b) {}}
      child {node {}}
    }
    child {node {}
      child {node {}}
      child {node {}}
    }
  };

\draw[cyan, line width=1mm] (u0b) -- (u1b);
\draw[cyan, line width=1mm] (u1b) -- (u2b);
\draw[cyan, line width=1mm] (u2b) -- (u3b);

\end{tikzpicture}
	\caption{In this picture, it shows an example of a forest $\tree_\freq$ where the orange nodes are the marked ones, and the blue path is the root-to-leaf path ending at the current input batch $F$. Upon the arrival of a new input batch $F$, we need to update the root-to-leaf tree path and the marked sets accordingly.}\label{fig:pre}
\end{figure}

\paragraph*{Coloring $F_{l, r}$.} To find colors for edges in $F_{l, r}$, we first need to find a proper palette for each vertex $v\in R$ such that $\deg_F(v)\in [2^r, 2^{r+1})$. We will describe an algorithm $\textsc{FreqVec}(F)$\label{alg:freqvec} which finds a proper frequency vector $\freq_v = (f_1, f_2, \ldots, f_h)$ for each such $v\in R$ and use the palette $\clr^{F}$ associated with leaf node $F$ in the forest $\tree_{\freq_v}$. To identify the proper frequency vector $\freq_v$, we will figure out each coordinate $f_1, f_2, \ldots, f_h$ one by one inductively. 
\begin{itemize}
	\item \textbf{Basis.} Let $N$ be the unique level-$1$ node containing $F$ (recall that level-$1$ nodes are defined irrespective of frequency vectors). Find $f_1 \in \{\Delta^\epsilon, 2\Delta^\epsilon, \ldots, 2^{r+1}\}$ which is the smallest integer such that there exists a frequency vector $\freq' = (f_1, f_2', f_3', \ldots)$, so that $M_{v, \freq'}$ contains less than $f_1$ children of $N$. Note that such a vector must exist, because $f_1 = 2^{r+1}$ always satisfies this requirement.
	
	If $f_1 = \Delta^\epsilon$, return the 1-dimensional vector $\freq_v = (f_1)$.
	
	\item \textbf{Induction.} In general, assume we have specified a prefix $f_1, f_2, \ldots, f_k$ for some $k\geq 1$. Note that all the forests $\tree_{\freq'}$ share the same topological structures from level $0$ up to $k+1$. Let $N$ be the unique level-$(k+1)$ node containing $F$ conditioning on $f_1, f_2, \ldots, f_k$. Check all the frequency vectors $\freq'$ that begin with the prefix $f_1, f_2, \ldots, f_k$, and find the smallest possible $f_{k+1}\in \{\Delta^\epsilon, 2\Delta^\epsilon, \ldots, f_k\}$ such that there exists a frequency vector $\freq' = (f_1, f_2, \ldots, f_k, f_{k+1}, f_{k+2}', \ldots )$ under the condition that $M_{v, \freq'}$ contains less than $f_{k+1}$ children of $N$. Note that such a vector must exist, because $f_{k+1} = f_k$ always satisfies this requirement.
	
	If $f_{k+1} = \Delta^\epsilon$, then return the vector $\freq_v = (f_1, f_2, \ldots, f_k, f_{k+1})$.
	
	\item \textbf{Choosing palette $\clr_v$.} Once we have finished the above inductive process, we need to choose a palette $\clr_v$ for $v$ which contains $25\cdot 2^{l+r+2}$ different colors. Basically, we will choose the palette associated with leaf node $F$ in forest $\tree_{\freq_v}$ as $\clr_v$, but must ensure that $\clr_v$ has not been used previously.
	
	To make sure of this, let $P$ denote the root-to-leaf tree path ending at leaf node $F$ in $\tree_{\freq_v}$. Travel down the path from root to leaf and enumerate all the encountered tree nodes. For every such tree node $N$, if $N$ already contains a sibling $W\in M_{v, \freq_v}$ and $\clr^W = \clr^N$, then abort this procedure and assign $\clr_v\leftarrow \emptyset$.
	
	In the end, if this procedure terminates without abortion, then assign $\clr_v\leftarrow \clr^F$; here $\clr^F$ refers to the palette of size $25\cdot 2^{l+r+2}$ associated with leaf node $F$ in forest $\tree_{\freq_v}$.
\end{itemize}

In this way, we can select a frequency vector $\freq_v$ for every vertex $v\in R$ such that $\deg_F(v)\in [2^r, 2^{r+1})$. Let $\clr_v$ be the palette assigned to leaf node $F$ by forest $\tree_{\freq_v}$. Next, we are going to color all edges in $F_{l, r}$ using colors from $\clr_v$ for edges incident on $v\in R$. Go over each vertex $u\in L$, and list its neighbors $v_1, v_2, \ldots, v_k, k<2^{l+1}$ in $F_{l, r}$. For any index $1\leq i\leq k$, if $\clr_{v_i}\neq \emptyset$, then reserve a tentative color to edge $(u, v_i)$, which is the $\kappa_i$-th color in palette $\clr_{v_i}$ and $$\kappa_i = 5\cdot 2^{l+1}\cdot\left(\cnt(\clr_{v_i}) + r_u\right) + i \mod 25\cdot 2^{l+r+2}$$
Recall that $\cnt(\clr_{v})$ counts the number of times that palette $\clr_v$ has appeared at previous leaf nodes under $\tree_\freq$. Notice that $k<2^{l+1}$, these tentative colors are different around $u$.

On the side of $v\in R$, it checks all the tentative colors on all of its edges in $F_{l, r}$. If a tentative color is used more than once, then it keeps only one of them, and turns other tentative colors back to $\bot$. Finally, for each edge $e\in F_{l, r}$, assign $e$ its own tentative color and print it in the output stream.

\paragraph*{Postprocessing Marked Sets.} After processing the input batch $F$, for every vertex $v\in R$ such that $\deg_F(v)\in [2^r, 2^{r+1})$, add node $F$ to $M_{v, \freq_v}$; note that we mark $F$ irrespective of whether $\clr_v$ is $\emptyset$ or not as will be important for establishing \Cref{inv}(4). The whole algorithm is summarized in \Cref{alg-low-deg}.

\begin{algorithm}
    \SetNoFillComment
    \caption{$\textsc{ColorLowDeg}(F)$}\label{alg-low-deg}
    \tcc{pre-processing marked sets}
    \For{$v\in R$ and frequency vector $\freq$}{
        call $\textsc{UpdateMarkSet}(F)$ (as described in \Cref{alg:UpdateMarkSet}) to remove marked nodes in previous tree components in $\tree_\freq$,\\
        and elevate the positions of all the marked nodes in $M_{v, \freq}$ to their ancestors which are children of $P$\;
    }
    \tcc{select frequency vector $\freq_v$ and palette $\clr_v$}
    \For{$v\in R$ such that $\deg_F(v)\in [2^r, 2^{r+1})$}{
        call $\textsc{FreqVec}(F)$ (as described in \Cref{alg:freqvec}) to determine the frequency vector $\freq_v = (f_1, f_2, \ldots )$ progressively\;
        find palette $\clr_v$ while ensuring no conflicts with sibling nodes\;
    }
    \tcc{assign colors for $F_{l, r}$}
    \For{$u\in L$}{
        let $v_1, v_2, \ldots, v_k$ be all of $u$'s neighbors in $F_{l, r}$\;
        assign edge $(u, v_i)$ a tentative color which is the $\kappa_i$-th color in palette $\clr_{v_i}$, where
        $\kappa_i = 5\cdot 2^{l+1}\cdot\left(\cnt(\clr_{v_i}) + r_u\right) + i \mod 25\cdot 2^{l+r+2}$\;
    }
    \For{$v\in R$ such that $\deg_F(v)\in [2^r, 2^{r+1})$}{
        discard tentative colors that appear more than once around $v$ in $F_{l, r}$\;
        output the unique tentative colors to the output stream\;
    }
    \tcc{post-processing marked sets}
    \For{$v\in R$ such that $\deg_F(v)\in [2^r, 2^{r+1})$}{
        add $F$ to $M_{v, \freq_v}$\;
    }
\end{algorithm}

\subsubsection{Proof of Correctness}
To begin with, let us first bound the total number of different colors that we could possibly use.
\begin{lemma}
	The total number of colors that the algorithm could use is $O((\log\Delta)^{O(1/\epsilon)}\Delta^{1+\epsilon}\cdot 2^l)$.
\end{lemma}
\begin{proof}
    By the design of the forest structures, for each frequency vector $\freq$, the number of colors assigned to each tree in the forest $\tree_\freq$ is bounded by $25 \cdot 2^{l + r + 2} \cdot \prod_{i = 1}^h (5f_i). $  Since each tree in $\tree_\freq$ covers $ 2^{r + 1} \cdot \prod_{i = 1}^{h - 1} f_i $ batches, and there are at most $m / n \leq \Delta$ batches, the total number of colors in $\tree_\freq$ is bounded by  
    $$ 
    25 \cdot 2^{l + r + 2} \cdot \prod_{i = 1}^h (5f_i) \cdot \lceil \frac{\Delta}{ 2^{r + 1} \cdot \prod_{i = 1}^{h - 1} f_i } \rceil
    = O(\Delta^{1 + \epsilon} \cdot 2^l \cdot 5^{O(1 / \epsilon)}).   
    $$    
    Since there are $(\log\Delta)^{O(1/\epsilon)}$ different choices for $\freq$, the total number of colors used by the algorithm is $O((\log\Delta)^{O(1/\epsilon)}\Delta^{1+\epsilon}\cdot 2^l). $  
\end{proof}

Let us next state some basic properties of any forest $\tree_{\freq}$. 
\begin{lemma}
	Each palette is used by at most $2^{r+1} / \Delta^\epsilon$ different batches.
\end{lemma}
\begin{proof}
	Consider a palette $\clr^N$ allocated to a leaf node $N$. By the construction of color packages:
	\begin{itemize}
		\item At the root, exactly one package contains this palette, and $ f_{h - 1} / \Delta^\epsilon $ children of the root inherit the palette. Therefore, at level $h - 1$, the palette $\clr^N$ is used in at most $ f_{h - 1} / \Delta^\epsilon $ nodes.
		\item Suppose at level \(k\), the palette \(\clr^N\) is used in \( f_k / \Delta^\epsilon \) nodes. Each of these nodes has \( f_{k-1} / f_k \) children that inherit the palette. Thus, at level \(k-1\), the palette is used in at most \( f_{k-1} / \Delta^\epsilon \) nodes.
	\end{itemize}
	By induction, the palette \(\clr^N\) is used in at most \( 2^{r+1} / \Delta^\epsilon \) different batches at level 0, as required. 
\end{proof}

\begin{corollary}\label{cnt-bound}
	During the algorithm, the values of the counters $\cnt(\clr)$ never exceed $2^{r+1} / \Delta^\epsilon$ for any palette $\clr$.
\end{corollary}

To bound the total space, it is clear that the bottleneck is storing all the marked sets \(M_{v, \freq}\), since all the forest structures and color assignments only require space \(O\brac{(\log\Delta)^{O(1/\epsilon)}\Delta}\).

\begin{lemma}
    If \Cref{inv} is satisfied, then at any point during the execution of the algorithm, for any given frequency vector $\freq$, the total size of marked sets $\sum_{v\in R}|M_{v, \freq}|$ is bounded by $O(2^h n)$. Consequently, the total space usage is bounded by $O((\log\Delta)^{O(1/\epsilon)} n)$.
\end{lemma}

\begin{proof}
    For any level-$k$ vertex $N \in \tree_\freq$, the subtree $ \tree_{\freq}(N) $ contains exactly $ 2^{r+1} \cdot \prod_{i=1}^{k-1} f_i $ batches, and thus $ 2^{r+1} \cdot \prod_{i=1}^{k-1} f_i \cdot n$ edges. 
    Let $U$ denote the union of all batches in $ \tree_{\freq}(N) $. For any vertex $v$, if $M_{v, \freq} \ni N$, by \Cref{inv}(2), we have $\deg_U(v) \geq 2^{r-k} \cdot \prod_{i=1}^k f_i$. 	
	Thus, the number of vertices $v$ such that $ M_{v, \freq} \ni N$ is bounded by: 
	$$
	\frac{2^{r+1} \cdot \prod_{i=1}^{k-1} f_i \cdot n}{2^{r-k} \cdot \prod_{i=1}^k f_i} = \frac{2^{k+1} \cdot n}{f_k}.
	$$

    Suppose the current working batch is $F$. By \Cref{inv}(1), any marked vertex is a child of a node on the root-to-leaf path $P$. Thus, among all vertices at level $k$, there are $f_k$ vertices that are candidates for being marked. 
	Taking the summation over all levels, the total size of $ \sum_{v \in R} |M_{v, \freq}| $ is bounded by
	$$
	\sum_{v \in R} |M_{v, \freq}|   
	\le \sum_{k = 0}^{h + 1}   \frac{2^{k + 1} \cdot n}{f_k} \cdot f_k    
	= O(2^h n). 
	$$
    Since the depth of the tree is at most $O(1/\epsilon)$, and there are $O( (\log \Delta)^{O(1 / \epsilon)})$ distinct frequency vectors $\freq$, the total space usage is therefore bounded by $O((\log\Delta)^{O(1/\epsilon)} n)$.
    \qedhere
\end{proof}
Now, our main focus will be on verifying all the properties in \Cref{inv}.
\begin{lemma}\label{verify-inv}
	\Cref{inv} is preserved by the algorithm throughout its execution.
\end{lemma}
\begin{proof}
	Property (1) is preserved in a straightforward manner by the way we maintain all the marked sets $M_{v, \freq}$ in the pre- and post-processing step for each input batch. Property (4) is also guaranteed because our updates to marked sets only depend on the input stream, not on the randomness used by the algorithm.
	
	As for property (3), according to the coloring algorithm, for any previous input batch $F'$ where $\deg_{F'}(v)\in [2^r, 2^{r+1})$, if vertex $v\in R$ used some colors associated with leaf node $F'$ in forest $\tree_{\freq}$, then the algorithm must have added $F'$ to $M_{v, \freq}$ back then. Then, according to the preprocessing rules, $F'$ would always be contained in the subtree of some marked as long as the current leaf node $F$ belongs to the same connected tree component as $F'$ in $\tree_{\freq}$.
	
	Next, we will focus on property (2). We will prove this by an induction on time. As the basis, property (2) holds trivially because all the marked sets are empty. Consider the arrival of any new input batch $F$ and any vertex $v\in R$ such that $\deg_F(v)\in [2^r, 2^{r+1})$. 
	
	First, we argue that the pre-processing step does not harm property (2). This is because the inequality in property (2) is required for all ancestors of any marked node (including itself), so if it held right before the arrival of $F$, then it should also hold after the pre-processing step as we are only raising the positions of marked nodes to their ancestors.
	
	Next, we consider the coloring step and the post-processing step. According to the coloring procedure, since $v$ will only use colors and mark nodes in forest $\tree_{\freq_v}$, we will only need to worry about property (2) in forest $\tree_{\freq_v}$. Before the post-processing step, let $P$ be the root-to-leaf path ending at leaf node $F$ in $\tree_{\freq_v}$, and let $W$ be the highest (in terms of levels) node on $P$ such that \(V(\tree_{\freq_v}(W))\cap M_{v, \freq_v} = \emptyset\), namely the subtree $\tree_{\freq_v}(W)$ does not contain any marked nodes; this node $W$ always exists because $F$ itself is a possible candidate.
	
	To prove property (2) after we add $F$ to $M_{v, \freq_v}$, we only need to verify the inequality for all nodes on $P$ between $W$ and $F$. List all these nodes as $F = N_0, N_1, N_2, \ldots, N_k = W$. We will prove the inequality of property (2) for each index $i = 0, 1, \ldots, k$.
	\begin{itemize}
		\item When $i = 0$, we already know that $\deg_F(v)\geq 2^r$, so the inequality holds.
		\item For any index $1\leq i\leq k$, consider the procedure that chose the coordinate $f_i$. Because of the minimality of $f_i$, we know that for some alternate frequency vector $\freq' = (f_1, f_2, \ldots, f_{i-1}, f_i / 2, \ldots)$, there are at least $f_i/2$ marked children of $N_i'$ (before the post-processing step), where $N_i'\in V(\tree_{\freq'})$ sits at the same topological position as $N_i$ (or in other words, $\tree_{\freq_v}(N_i)$ and $\tree_{\freq'}(N_i')$ are isomorphic). Let $U$ be the union of the batches which are all leaf nodes of $\tree_{\freq'}(N_i')$. Therefore, by our inductive assumption regarding property (2), we know that: 
		$$\deg_U(v)\geq  (f_i / 2)\cdot 2^{r - i+1}\cdot \prod_{j=1}^{i-1}f_j = 2^{r - i}\cdot \prod_{j=1}^{i}f_j$$
		This verifies the inequality at node $N_i$.\qedhere
	\end{itemize}
\end{proof}

Next, let us turn to the color assignment part. First, we need to verify that the algorithm never assigns the same color twice around the neighborhood of a single vertex.
\begin{lemma}
	In the output stream, the algorithm never prints the same color for two adjacent edges.
\end{lemma}
\begin{proof}
	First, let us rule out color conflicts for edges around the same vertex $u\in L$. This is rather straightforward, because according to the algorithm description, within each batch, we only assign tentative colors with distinct color indices around each vertex $u\in L$. For two different batches $F, F'$, if we happen to use the same color palette $\clr$ in $F, F'$ around the same vertex $u$, then the counter values $\cnt(\clr)$ must be different in these two batches. According to \Cref{cnt-bound} and that $r_u\in [3\cdot 2^{r+1}]$, the value of $\cnt(\clr_v+r_u)$ is always in the range $[3\cdot 2^{r+1}, 4\cdot 2^{r+1}]$. Together with the fact that $\deg_F(u), \deg_{F'}(u) < 2^{l+1}$, the algorithm must use distinct colors in $F, F'$ from $\clr$.
	
	Next, let us rule out color conflicts for edges around the same vertex $v\in R$. For any color palette provided by $\tree_{\freq_v}$ which was used before in some previous batch $F'$, according to \Cref{inv}(3), there must be a marked node $N'$ whose subtree contains $F'$. According to our coloring procedure, $\clr_v$ is nonempty only when $\clr^{N}$ is disjoint from all the color packages of its marked siblings $N'$, for any node $N$ on the root-to-leaf path to $F$ in $\tree_{\freq_v}$. Therefore, $v$ could not reuse any  palettes previously assigned. Furthermore, since we discard repeated colors within a single palette, the assigned colors must be distinct.
\end{proof}

Finally, let us prove that the algorithm successfully colors a good fraction of all edges in the input stream.
\begin{lemma}
	The total number of colored edges in the output stream is at least $\delta m$ in expectation, where $\delta = 2^{-O(1/\epsilon)}$.
\end{lemma}
\begin{proof}
	Fix any input batch $F$ and any vertex $v\in R$ such that $\deg_F(v)\in [2^r, 2^{r+1})$, it suffices to lower bound the expected number of colored edges in $F_{l, r}$ incident on $v$.
	
	First, we need to analyze the probability that the color palette $\clr_v$ is nonempty, based on the randomness of the distribution of color packages on forest $\tree_{\freq_v}$. As before, let $P$ denote the root-to-leaf path in $\tree_{\freq_v}$ ending at $F$, and let $F = N_0, N_1, \ldots, N_h$ be all the nodes on the tree path. According to the selection of the frequency vector $\freq_v$, for any $0\leq k\leq h-1$, $N_k$ has less than $f_{k + 1}$ marked siblings. By design, the color packages of $N_{k + 1}$ are determined by the length-$f_{k}$ prefix of a random permutation $(f_{k} / f_{k + 1})\odot [5f_{k + 1}]$. Therefore, by the independence guarantee from \Cref{inv}(4), the probability that $\clr^{N_k}$ does not conflict with the color packages of any other marked siblings is at least $(1 - 1/(5f_{k + 1}))^{f_{k + 1}} \ge  4/5$. By applying this bound over all levels, the probability that $\clr_v$ is nonempty is at least $(4/5)^{h} \geq (4/5)^{1/\epsilon}$.
	
	Next, conditioning on the event that $\clr_v\neq \emptyset$, let us analyze the amount of edges in $F_{l, r}$ that are colored around $v$. Let $u_1, u_2, \ldots, u_k$ be the neighbors of $v$ in graph $(V, F_{l, r})$, where $k < 2^{r+1}$. 
	For any fixed $1\leq i\leq k$, as $r_{u_i}$ was chosen uniformly at random from $[3\cdot 2^{r+1}]$, the probability that $r_{u_i}\neq r_{u_j}$ for all $j\neq i$ is at least $(1 - 1 / (3 \cdot 2^{r + 1}))^k \ge 2/3$. This ensures that the tentative colors between $u_i$ and $v$ survive in the output stream with probability at least $2/3$.
	
	Overall, the expected number of colored edges in $F_{l, r}$ would be at least $(2/3)\cdot (4/5)^{1/\epsilon}\cdot |F_{l, r}|$, which concludes the proof.
\end{proof}

\subsection{Proof of \Cref{high}} \label{subsec:high}
Without loss of generality, assume $\Delta$ is a power of 2. 
\subsubsection{Data Structures}  
At the beginning, for each vertex $u \in L$, draw a random number $s_u \in [\Delta / 2^l]$ uniformly at random. For each vertex $v \in R$, draw a random number $t_v \in [\Delta / 2^r]$ uniformly at random. 

As in the preliminary steps, the input stream is divided into batches of size $n$ (except for the last one). For each $u\in L$, let $\cnt(u)$ count the number of previous batches $F$ where $\deg_{F}(u)\in [2^l, 2^{l+1})$, and symmetrically let $\cnt(v)$ count the number of previous batches $F$ where $\deg_F(v)\in [2^r, 2^{r+1})$ for $v\in R$.

Additionally, we will use a palette matrix $\mat$ of size $\frac{\Delta}{2^l} \times \frac{\Delta}{2^r}$, where each entry in $\mat[i,j]$ corresponds to a distinct palette of size $\Delta_0$, with $\Delta_0 \triangleq \left\lceil 4\cdot \left( \frac{2^{l+r+1}}{\Delta} + 1 \right) \right\rceil$. All the colors can be represented as integers in the range $\left[\frac{\Delta_0\cdot\Delta^2}{2^{l+r}}\right]$ naturally.

\subsubsection{Algorithm Description}
Let us describe the coloring procedure upon the arrival of a new input batch $F$. For each vertex $u \in L$, propose a tentative row index $x_u = (s_u + \cnt(u)) \mod \frac{\Delta}{2^l}$. For each vertex $v \in R$, propose a tentative column index $y_v = (t_v + \cnt(v)) \mod \frac{\Delta}{2^r}$. For each edge $(u, v) \in F$, we will assign a color from the matrix $\mat[x_u, y_v]$ in the following manner.

Let $E_{x, y}$ be the set of edges $(u, v) \in F_{l, r}$ where $(x_u, y_v) = (x, y)$, and let $G_{x, y}$ be the subgraph whose edge set is $E_{x, y}$. To color $G_{x, y}$ with only $\Delta_0$ colors, we need to prune it so that its maximum degree does not exceed $\Delta_0$, which is done in this way: for each edge $(u, v) \in E_{x, y}$, if $\max\{\deg_{E_{x, y}}(u), \deg_{E_{x, y}}(v)\} > \Delta_0$, mark it as $\bot$ (uncolored) and remove it from $G_{x, y}$. Finally, since $G_{x, y}$ is a bipartite graph, we can apply the $\Delta_0$-edge coloring algorithm \cite{cole2001edge} to color $G_{x, y}$ using the palette $\mat[x, y]$ which has size $\Delta_0$.

Finally, for each vertex $u\in L$ such that $\deg_F(u)\in [2^l, 2^{l+1})$, increment the counter $\cnt(u)$ by $1$; also, increment the counters for $v\in R$ in a symmetric way. The whole algorithm is summarized in \Cref{alg-high-deg}.

\begin{algorithm}
    \SetNoFillComment
    \caption{$\textsc{ColorHighDeg}(F)$}\label{alg-high-deg}
    define $x_u \leftarrow s_u + \cnt(u) \mod \Delta / 2^l, \forall u\in L$\;
    define $y_v \leftarrow t_v + \cnt(v) \mod \Delta / 2^r, \forall v \in R$\;
    define $E_{x, y} = \{(u, v)\in F_{l, r}\mid (x_u, y_v) = (x, y)\}, \forall (x, y)$\;
    \For{every pair $(x, y)$ and edge $(u, v)\in E_{x, y}$}{
        \tcc{prune $G_{x,y}$ to cap its maximum degree}
        remove edge $(u, v)$ from $E_{x, y}$ if $\max\{\deg_{E_{x, y}}(u), \deg_{E_{x, y}}(v)\} > \Delta_0$\;
    }
    \For{every pair $(x, y)$}{
        use palette $\mat[x, y]$ to color $E_{x, y}$\;
    }
    increment counters $\cnt(*)$\;
\end{algorithm}

\subsubsection{Proof of Correctness}

\paragraph*{Proper Coloring.}  
To prove that the colored edges form a proper coloring, we need to show that for any two distinct edges $e_1 = (u, v_1)$ and $e_2 = (u, v_2)$ sharing a common vertex $u$, their colors are different. Since the algorithm is symmetric for $L$ and $R$, we can assume $u\in L$. There are several cases below.

\begin{itemize}
	\item If $e_1$ and $e_2$ use different matrix entries in $\mat$, their colors are already distinct.
	\item Suppose both edges use palette $\mat[x, y]$. Then, there are two sub-cases below.
	\begin{itemize}
		\item If $e_1$ and $e_2$ belong to different batches $F_1$ and $F_2$ with batch counters $\cnt^{(1)}(u)$ and $\cnt^{(2)}(u)$, we have $x \equiv s_u + \cnt^{(1)}(u) \equiv s_u + \cnt^{(2)}(u) \pmod {\Delta / 2^l}$. Since $\deg_F(u) \in [2^l, 2^{l + 1})$, $\cnt_u$ never exceeds $\Delta / 2^l$. Thus, $\cnt^{(1)}(u) = \cnt^{(2)}(u)$, which leads to a contradiction.
		\item If $e_1$ and $e_2$ belong to the same batch, they belong to the same subgraph $G_{x, y}$. The correctness of the offline coloring algorithm guarantees they get distinct colors.
	\end{itemize}
\end{itemize}

\paragraph*{Space Usage.}  
For each vertex $u$, we maintain its batch counter $\cnt_u$ and random shift $s_u$ for $u \in L$ (or $t_v$ for $v \in R$), requiring $O(n)$ space in total. During the batch coloring process, we also store the indices $x_u$ and $y_v$, which require additional $O(n)$ space. Furthermore, each subgraph $G_{x, y}$ has at most $n$ edges, so the offline coloring algorithm requires $O(n)$ space. These spaces are reused across different subgraph coloring processes and different batches. Therefore, the overall space complexity is $O(n)$.

\paragraph*{Number of Colors.}  
The total number of colors is given by
$$
\frac{\Delta}{2^l} \cdot \frac{\Delta}{2^r} \cdot \Delta_0 = O \left( \Delta + \frac{\Delta^2}{2^{l+r}} \right).
$$

Next, we show that at least half of the edges get colored in expectation.

\begin{lemma}\label{lemma:high-deg-uncolored}
	During the algorithm, at least a $1/2$ fraction of the edges are colored in expectation.
\end{lemma}
\begin{proof}
    For any edge $(u, v)\in F_{l, r}$ such that $u\in L, v\in R$, we estimate the probability that $(u, v)$ is colored. 
    Define $x = x_u, y = y_v$. It suffices to lower bound the probability that $\deg_{E_{x, y}}(u), \deg_{E_{x, y}}(v) \leq \Delta_0$.

    Let $(u, v_1), (u, v_2), \ldots, (u, v_k)\in F_{l, r}, k<2^{l+1}-1$ be all edges incident on $u$ other than $(u, v)$. Since $G$ is a simple graph, all the vertices $v_1, v_2, \ldots, v_k$ are distinct and are different from $v$. Since $y_{v_i}$ is uniformly distributed in $[\Delta / 2^r]$, the probability that $y_{v_i} = y$ is at most $2^r / \Delta$. Using Markov's inequality, we have:
    $$\Pr[\deg_{E_{x, y}}(u)\geq \Delta_0]\leq \frac{(2^{l+1}-2)\cdot 2^r / \Delta}{\Delta_0} \leq 1/4$$
    Symmetrically, we can argue that $\Pr[\deg_{E_{x, y}}(v)\geq \Delta_0] \leq 1/4$. Hence, the probability that $(u, v)$ remains in $E_{x, y}$ after the pruning procedure would be at least $1/2$. This concludes the proof.
\end{proof}

\section{Derandomization via Bipartite Expanders}

In this section we will de-randomize \Cref{rand} and prove \Cref{main}. By the reductions from the preliminary section, it suffices to prove the following two statements.

\begin{lemma}\label{det-low}
	Fix a parameter $\epsilon > 0$. Given a graph $G = (V, E)$ on $n$ vertices with maximum vertex degree $\Delta$, for any constant $\epsilon > 0$, and fix an integer pair $(l, r)$, there is a deterministic W-streaming algorithm that outputs a coloring of all edges in $F_{l, r}$. The algorithm uses $O\brac{(\log\Delta)^{O(1 / \epsilon)}\cdot(1/\epsilon)^{O(1/\epsilon^3)}\cdot \Delta^{1+2\epsilon}\cdot 2^l}$ colors and $O\brac{n\cdot (\log n)^{O(1 / \epsilon^4)}}$ space.
\end{lemma}

\begin{lemma}\label{det-high}
	Fix a parameter $\epsilon > 0$. Given a graph $G = (V, E)$ on $n$ vertices with maximum vertex degree $\Delta$, and fix an integer pair $(l, r)$, there is a deterministic W-streaming algorithm that outputs a coloring of all edges in $F_{l, r}$. The algorithm uses $O\brac{\Delta^{1+\epsilon} + \Delta^{2+\epsilon} / 2^{l+r}}$ colors and $O\brac{n \cdot (\log n)^{{O(1/\epsilon^3)}}}$ space.
\end{lemma}

\begin{proof}[Proof of \Cref{main}]
	Basically, \Cref{det-low} deals with the case when $\min\{l, r\}$ is small, and \Cref{high} deals with the case when $\min\{l, r\}$ is large. For any $(l, r)$, if $\min\{2^l, 2^r\}\leq \Delta^{1/3}$, then by \Cref{det-low}, the number of colors is at most $O\brac{(\log\Delta)^{O(1 / \epsilon)}\cdot(1/\epsilon)^{O(1/\epsilon^3)}\cdot \Delta^{4/3+\epsilon}}$ (after scaling down $\epsilon$ by $2$), and otherwise by \Cref{det-high} the number of colors would be $O\brac{\Delta^{4/3+\epsilon}}$. Either way, the total number of colors over all $(l, r)$ would can be bounded by $O\brac{(\log\Delta)^{O(1 / \epsilon)}\cdot(1/\epsilon)^{O(1/\epsilon^3)}\cdot \Delta^{4/3+\epsilon}}$.
\end{proof}

\subsection{Bipartite Expanders via Multiplicity Codes}
We will use explicit constructions of unbalanced bipartite expanders from \cite{kalev2022unbalanced} for de-randomization. For \Cref{det-low} we will only use bipartite expanders in a black-box manner, but for \Cref{det-high} we will also have to leverage some properties of the multiplicity code itself. 

\begin{definition}[bipartite expanders]
	Given a bipartite graph $H = (A\cup B, I)$, $H$ is a $(K, D)$-expander if for every $S\subseteq A$ of size $K$, the number of different neighbors of $S$ in $B$ is at least $K\cdot D$.
\end{definition}

Let $q$ be a prime number and $\field_q$ be the corresponding finite field. Let $a, b$ be two integers. Any vector in $\field_q^a$ can be interpreted as a uni-variate polynomial over $\field_q$, that is, as an element $f\in \field_q^{<a}[X]$. Define a mapping $\Gamma: \field_q^a\times \field_q\rightarrow \field_q^{b+2}$ as $\Gamma(f, x) = \left(x, f(x), f^{(1)}(x), \ldots, f^{(b)}(x)\right)$, where $f^{(i)}$ is the $i$-th iterated derivative of $f$ in $\field_q[X]$. 

\begin{lemma}[multiplicity codes are bipartite expanders \cite{kalev2022unbalanced}]\label{multi}
	Construct a bipartite graph $H = (A\cup B, I)$ where $A = \field_q^a, B = \field_q^{b+2}$, and for each vertex in $A$ which is represented as a uni-variate polynomial $f\in \field_q^{<a}[X]$, connect vertex $f$ to neighbors $\Gamma(f, x)\in B, \forall x\in \field_q$. Then, under the condition that $15\leq b+2\leq a\leq q$, $H$ is a $(K, D)$-expander for every $K>0$ and $D = q - \frac{a(b+2)}{2}\cdot (qK)^{\frac{1}{b+2}}$.
\end{lemma}

We will also need an upper bound on the degrees on the right-hand side of the bipartite expanders.
\begin{lemma}\label{deg}
	Let $H = (A\cup B, I)$ be the above bipartite graph by multiplicity codes. Then, for any $v\in B$, $\deg_H(v)\leq q^{a-b-1}$. 
\end{lemma}
\begin{proof}
	It suffices to show that for any vector $(x, x_0, x_1, \ldots, x_b)\in \field_q^{b+2}$, there exists at most $q^{a-b-1}$ polynomials in $\field_q^{<a}[X]$ such that $\left(f(x), f^{(1)}(x), \ldots, f^{(b)}(x)\right) = (x_0, x_1, \ldots, x_b)$. Since this equation is a full-rank linear equation on $\field_q$ over the coefficients of $f$ with $a$ variables and $b+1$ equality constraints, the number of solutions is at most $q^{a-b-1}$.
\end{proof}

\paragraph*{Online Perfect Matching in Bipartite Expanders.}
For the convenience of algorithm description, let us first extract a building block of our main algorithm here which is called \emph{online perfect matching in bipartite expanders}. In this problem, we can set up a bipartite expander $H = (A\cup B, I)$ of our own, and then an adversary picks a sequence of vertices $u_1, u_2, \ldots, u_K$ one by one. Every time a vertex $u_i$ is revealed to the algorithm, we need to irrevocably match $u_i$ to a vertex $v_i$ not previously matched, and we wish to make $K$ as large as possible.

\begin{lemma}\label{perfect}
	For any integer parameters $a, b, q$ such that $15\leq b+2\leq a\leq q$, there is a deterministic construction of a bipartite graph $H = (A\cup B, I)$, $|A| = q^a, |B|= q^{b+2}\cdot \ceil{(b+2)\log_2 q}$, such that the online perfect matching on $H$ can be solved as long as the number of online vertices in $A$ is at most $K\leq \frac{(2q-2)^{b+2}}{q\cdot a^{b+2}(b+2)^{b+2}}$. Plus, the algorithm only takes space $O(K)$.
\end{lemma}
\begin{proof}
	To do this, we will make multiple copies of the bipartite expander based on multiplicity codes in \Cref{multi}. More specifically, fix any parameter $a, b$ and prime $q$ such that $15\leq b+2\leq a\leq q$, for every index $1\leq j\leq \lceil (b+2)\log_2 q\rceil$, construct a bipartite expander $H_j = (A\cup B_j, I_j)$ where $A = \Gamma_q^a, B = \Gamma_q^{b+2}$ using the construction from \Cref{multi}. After that, define $B = \bigcup_{j=1}^{\lceil (b+2)\log_2 q\rceil} B_j$, $I = \bigcup_{j=1}^{\lceil (b+2)\log_2 q\rceil} I_j$, and $H = (A\cup B, I)$.
	
	To run an online perfect matching on graph $H$, we use the simple greedy matching approach. Basically, for each upcoming vertex $u_i\in A$ selected by the adversary, we find the smallest index $1\leq j\leq \lceil (b+2)\log_2 q\rceil$ such that there exists a neighbor of $u_i$ in $H_j$ not matched with previous online vertices $u_1, \ldots, u_{i-1}$. We show that this approach can accommodate a large number of online vertices in a perfect matching.
	
	\begin{claim}
		If $K \leq \frac{(2q-2)^{b+2}}{q\cdot a^{b+2}(b+2)^{b+2}}$, then all vertices $u_1, \ldots, u_K$ are matched in $H$ by the greedy algorithm.
	\end{claim}
	\begin{proof}[Proof of claim]
		It suffices to prove by an induction that the number of vertices in $u_1, u_2, \ldots, u_K$ not matched in $\bigcup_{i=1}^j H_j$ is at most $K / 2^j$. This is proved by an induction on $j\geq 0$ ($H_0 = \emptyset$). The basis $j = 0$ holds trivially. In general, suppose $u_1', u_2', \ldots u_k'$ is a sub-sequence of $u_1, u_2, \ldots, u_K$ which are not matched in $\bigcup_{i=1}^{j-1}H_i$. According to \Cref{multi}, for any $k\leq K \leq \frac{(2q-2)^{b+2}}{q\cdot a^{b+2}(b+2)^{b+2}}$, the graph $H_j$ is a $(k, D)$-expander where $D = q - \frac{a(b+2)}{2}\cdot (qK)^{\frac{1}{b+2}} \geq q - (q-1) = 1$. Therefore, according to Hall's theorem, there exists a perfect matching which matches all vertices $u_1', u_2', \ldots u_k'$ in $H_j$. Since greedy matching is a $1/2$-approximation of maximum matching, the algorithm must have matched at least $k/2$ of $u_1', u_2', \ldots u_k'$ in $H_j$. By induction, there are at most $k/2 \leq K/2^j$ vertices which are not matched in $\bigcup_{i=1}^j H_j$.
	\end{proof}
\end{proof}

\subsection{Proof of \Cref{det-low}}
Let $\delta = \epsilon^2/10$ be a small constant. For convenience, let us assume $1/\delta$ is an integer. If $\Delta < \log^{20/\epsilon^2\delta} n$, then a $O(n\log^{200/\epsilon^4}n)$-space algorithm can store the whole input graph in memory and print a $\Delta$-edge coloring in the output stream. For the rest, let us assume $\Delta \geq \log^{20/\epsilon^2\delta}n$.

\subsubsection{Data Structures}
We will reuse all the forest data structures defined previously in \Cref{Btree}, so we will not describe them again here. However, the color package allocations on those forests will be different because previously we have used randomization for this part.

\paragraph*{Deterministic Color Allocation on Forests.} For any choice of the frequency vector $\freq = (f_1, f_2, \ldots, f_h)$, we will allocate color packages at each tree node of each forest $\tree_\freq$ in a deterministic manner. Take a prime number $q\in [\Delta^\delta, 2\Delta^\delta)$. By construction, each forest structure $\tree_\freq$ is a forest of $h+1$ levels (from level-$0$ to level-$h$), and each tree is rooted at a level-$h$ node. For each coordinate $f_i$ of the frequency vector $\freq = (f_1, f_2, \ldots, f_h)$, define some parameters below:
$$\begin{aligned}
	\lambda &= \ceil{\log_2^{2+3/\delta}n\cdot (2 + 3/\delta)^{2+3/\delta}}\\
	b_0 &= \left\lceil\log_q(\lambda\cdot 2^{r+1})\right\rceil\\
	b_i &= \left\lceil\log_q(\lambda\cdot f_i)\right\rceil, \forall 1\leq i\leq h\\
	\lambda_i &= \lambda\cdot \ceil{(b_i+2)\cdot\log_2q}, \forall 1\leq i\leq h
\end{aligned}$$

Here are some basic estimations of these parameters.
\begin{lemma}\label{param-ineq}
    For any $0\leq i\leq h$, we have $b_i\leq 3/\delta$.
\end{lemma}
\begin{proof}
    By the assumption that $\Delta \geq \log^{20/\epsilon^2\delta}n$, the following inequality holds when $n$ is a super-constant:$$\lambda = \ceil{\log_2^{2+3/\delta}n\cdot (2 + 3/\delta)^{2+3/\delta}} < \Delta$$
    Hence, $b_i = \ceil{\log_q(\lambda \cdot f_i)}\leq \ceil{\log_q \Delta^2} < 3/\delta$ as $q\in [\Delta^\delta, 2\Delta^\delta)$.
\end{proof}

Allocate an overall color package $\clr^\freq$ with $2^{l+1}\cdot q^{b_0+3}\cdot\prod_{i=1}^h(\lambda_i\cdot q^{b_i+2})$ new colors. Each color in this color package can be identified as a tuple $(c_0, c_1, \ldots, c_h)$ from the direct product space $[2^{l+1}\cdot q^{b_0+3}]\times [\lambda_1\cdot q^{b_1+2}] \times\cdots\times [\lambda_{h}\cdot q^{b_{h}+2}]$. For any tuple $(c_1, c_2, \ldots, c_h)$, the collection of all colors $\{(*, c_1, c_2, \ldots, c_h)\}$ will be called a \textbf{palette}. 

As we did in the randomized algorithm, we will specify a color package for each tree node. To do this, for any level $0\leq i\leq h-1$, apply \Cref{perfect} which deterministically builds a bipartite graph $H_i^\freq = (A_i^\freq\cup B_i^\freq, I_i^\freq)$ where $|A_i^\freq| = q^{\ceil{\log_q f_i}}, |B_i^\freq| = q^{b_{i+1}+2}\cdot \ceil{(b_{i+1}+2)\cdot \log_2q}$.

\begin{lemma}\label{verify-perfect}
    Bipartite graph $H_i^\freq$ admits an online perfect matching with $f_{i+1}$ vertices in $A_i^\freq$.
\end{lemma}
\begin{proof}
    Plugging in the parameters in \Cref{perfect} and using \Cref{param-ineq}, the greedy algorithm always finds a perfect matching in $H_i^\freq$ as long the number of online vertices is at most 
    $$\begin{aligned}
        \frac{(2q-2)^{b_{i+1}+2}}{q\cdot \ceil{\log_q f_i}^{b_{i+1}+2}\cdot (b_{i+1}+2)^{b_{i+1}+2}}\geq \frac{\lambda f_{i+1}\cdot q}{q\cdot \log n^{2+3/\delta}\cdot (2+3/\delta)^{2+3/\delta}} \geq f_{i+1}
    \end{aligned}$$
    which concludes the proof.
\end{proof}

For any tree node $N$ in $\tree_\freq$ on level-$i$ for some $0\leq i\leq h-1$, let $P$ be the tree path which connects $N$ and the root of $\tree_{\freq}$. List all the nodes of $P$ as $N = N_i, N_{i+1}, \ldots, N_h$, and assume $N_j$ is the $k_j$-th child of $N_{j+1}$. Then, define the color package at $N$ to be:
$$\clr^N = \{(*, *, \ldots, *, c_{i+1}, c_{i+2}, \ldots, c_h)\mid (k_j, c_{j+1})\in I_j, \forall i\leq j < h\}$$
To justify this definition, since $1\leq k_j\leq f_j\leq |A_j^\freq|$ and $1\leq c_{j+1}\leq \lambda_{j+1}\cdot q^{b_{j+1}+2}\leq |B_j^\freq|$, we can encode $k_j$ as a vertex in $A_j^\freq$ and $c_{j+1}$ as a vertex in $B_j^\freq$.

Notice that, by definition, the same palette may appear at multiple leaf nodes (which represent input batches). For any leaf node $N$ (or equivalently, a batch), let $\cnt(\clr^N)$ count the total number of times that palette $\clr^N$ was contained in the color packages of previous leaf nodes (batches). Note that these counters do not require extra space, because we can recompute them upon the arrival of any input batch.

\paragraph*{Vertex-Wise Data Structures.} As before, we will maintain some data structures for the vertices in $V$. For the vertices in $L$, we will not maintain the random shifts $\{r_u\in [3\cdot 2^{r+1}] \mid u\in L\}$. Instead, we build a bipartite expander $H^\freq = (A^\freq\cup B^\freq, I^\freq)$ using multiplicity codes according to \Cref{multi} where $|A^\freq| = q^{\ceil{\log_q n}}, |B^\freq| = q^{b+2}, b = \ceil{\log_q(\lambda\cdot 2^{r+1})}$; recall the definitions of $q, \lambda$ from the previous paragraph.

For vertices $v\in R$, we will reuse the same data structures of marked nodes $M_{v, \freq}\subseteq V(\tree_\freq)$. In addition, for each marked node $N\in M_{v, \freq}$ which is on level-$i$, we will store a length-$h$ tuple $$\overrightarrow{c_N} = (*, *, \ldots, *, c^N_{i+1}, c^N_{i+2}, \ldots, c^N_h)$$ to represent all the color palettes that were used at node $N$ (the first $i$ coordinates are $*$). Similar to \Cref{inv}, the marked sets will have the following requirements.
\begin{invariant}\label{det-inv}
	We will ensure the following properties with respect to the marked nodes $M_{v, \freq}$ throughout the execution of the streaming algorithm.
	\begin{enumerate}[(1)]
		\item All nodes in $M_{v, \freq}$ are incomparable in forest $\tree_{\freq}$. Furthermore, suppose that the current input batch corresponds to a leaf $F$, and let $P$ tree path from $F$ to the tree root. Then, any node $N\in M_{v, \freq}$ is a child of a node on the root-to-leaf path $P$.
		
		\item For any node $N\in \tree_\freq$ on level-$k$ such that $M_{v, \freq}\cap V(\tree_\freq(N))\neq \emptyset$ , let $F_1, F_2, \ldots, F_s\subseteq E$ be all the input batches which correspond to leaf nodes in subtree $\tree_\freq(N)$. Take the union of batches $U = F_1\cup F_2\cup\cdots \cup F_s$. Then, we have $\deg_U(v)\geq 2^{r-k}\cdot\prod_{i = 1}^k f_i$.
		
		\item For any previous input batch $F'$ before $F$ such that:
		\begin{itemize}
			\item $F$ and $F'$ are in the same connected component in $\tree_\freq$,
			\item $\deg_{F'}(v)\in [2^r, r^{r+1})$,
			\item $v$ used some colors in $\clr^{F'}$ during the algorithm,
		\end{itemize} 
		we guarantee that $F'$ must be contained in some subtree $\tree_{\freq}(N)$ for some $N\in M_{v, \freq}$. In addition, assume $N$ is on level-$i$ in $\tree_\freq$, then we require that all colors used by the edges in $F'_{l, r}$ incident on $v$ share the suffix $\left(*, *, \ldots, *, c^N_{i+1}, c^N_{i+2}, \ldots c^N_h\right)$.
		
		\item For any pair of marked nodes $N_1, N_2\in M_{v, \freq}$, suppose their lowest ancestor is on level-$k$, then both tuples $\overrightarrow{c_{N_1}}, \overrightarrow{c_{N_2}}$ share the same suffix up to length $h-k$.
		
		In addition, for all marked nodes $N_1, N_2, \ldots, N_i$ which shares the same parent $W$ on level-$k$, suppose $N_j$ is the $k_j$-th child of $N$. We will make sure that $i\leq f_k$, and all the pairs $\left(k_j, c_{k}^{N_j}\right), 1\leq j\leq i$ form a perfect matching in the bipartite expander $H_{k-1}$ which is chosen by the online perfect matching algorithm we described in \Cref{perfect}.
	\end{enumerate}
\end{invariant}

\subsubsection{Algorithm Description}
Similar to the randomized algorithm in \Cref{rand-alg}, the deterministic also consists of three steps for each input batch $F$: preprocessing marked sets, coloring $F_{l, r}$, and post-processing marked sets.

\paragraph*{Preprocessing Marked Sets.} Since the current input batch has changed due to the new arrival $F$, we have might violated \Cref{inv}(2) as the root-to-leaf tree path may have changed. Therefore, we first need to update all the marked sets as in the following procedure which is named $\textsc{DetUpdateMarkSet}(F)$\label{alg:detupdatemarkset}.

Go over every vertex $v\in R$ and every frequency vector $\freq$. Consider the position of $F$ in the forest $\tree_\freq$, and let $P$ be the root-to-leaf path in $\tree_{\freq}$ ending at leaf $F$. First, remove all marked nodes $N\in M_{v, \freq}$ which are not in the same tree as $P$; note that forest $\tree_\freq$ is actually a forest of trees, and this may happen when $F$ is the first leaf in a new tree of $\tree_{\freq}$.

Next, go over every node $W$ on lying on the tree path $P$. Assume $W$ is on level-$k$. For any child node $N$ of $W$, if (1) $V(\tree_\freq(N)) \not\ni F$ and (2) $V(\tree_\freq(N))\cap M_{v, \freq} \neq \emptyset$, then take an arbitrary node $U\in V(\tree_\freq(N))\cap M_{v, \freq}$, assign $\overrightarrow{c_W} = \left(*, *, \ldots, *, c_{k}^U, c_{k+1}^U, \ldots, c_h^U\right)$. After that, remove all nodes in $V(\tree_\freq(N))\cap M_{v, \freq}$ from $M_{v, \freq}$ and add $N$ to $M_{v, \freq}$. 

\paragraph*{Coloring $F_{l, r}$.} The procedure $\textsc{DetFreqVec}(F)$\label{alg:detfreqvec} for selecting a frequency vector $\freq_v$ for each $v\in R$ such that $\deg_F(v)\in [2^r, 2^{r+1})$ is the same as \Cref{rand-alg} which is a deterministic subroutine. The different parts will be selection of the color palette $\clr_v$ and color assignment.
\begin{itemize}
	\item \textbf{Selecting $\clr_v$.} Let $W$ be the lowest ancestor of $F$ such that $V(\tree_{\freq_v}(W))\cap M_{v, \freq_v}\neq\emptyset$. If such a node $W$ does not exist, then $M_{v, \freq_v}$ must be empty at the moment. In this case, set $W$ to be the root of the tree containing $F$.
	
	In general, assume $W$ is on level-$k$, and define $$F = W_0, W_1, \ldots, W_{k-1}, W_k = W$$ to be the sub-path of $P$ connecting $W$ and $F$. Assume $W_{k-1}$ is the $s$-th child of $W$.
	
	According to our selection of $\freq_v$, $W$ has less than $f_k$ marked children. Assume these marked children of $W$ are $N_1, N_2, \ldots, N_t, s<f_k$, and $N_i$ is the $s_i$-th child of $W$. Under \Cref{det-inv}(4), all the pairs $\left(s_i, c_k^{N_i}\right)$ form a matching in bipartite expander $H_{k-1}^{\freq_v}$ and was selected by the online perfect matching algorithm from \Cref{perfect}. Since $t < f_k$, by \Cref{verify-perfect}, we can continue to apply \Cref{perfect} to find a value $c_k\notin \{c_k^{N_i} \mid 1\leq i\leq t\}$ such that $(s, c_k)\in I_{k-1}^{\freq_v} = E(H_{k-1}^{\freq_v})$.
	
	As for coordinates $c_{k+1}, c_{k+2}, \ldots, c_h$, assign $c_j = c_j^{N_1}, k<j\leq h$. To complete the coordinates of palette $\clr_v$, we still need to specify $c_1, c_2, \ldots, c_{k-1}$. To do this, for each level $0\leq i < k-1$, initialize an online perfect matching procedure on graph $H_i^{\freq_v}$. Assume $W_i$ is the $t_i$-th child of $W_{i+1}$, then view $t_i$ as the first online vertex in $H_i^{\freq_v}$ and match it using edge $(t_i, c_i)$ in $H_i^{\freq_v}$. In the end, set $\clr_v = (c_1, c_2, \ldots, c_h)$ be the palette we will use.
	
	\item \textbf{Color assignment.} Next, we are going to color all edges in $F_{l, r}$ using colors from $\clr_v$ for edges incident on $v\in R$. For any vertex $v\in R$, let $u_1, u_2, \ldots, u_k, k<2^{r+1}$ be all the neighbors of $v$ in $F_{l, r}$.
	Recall that $H^{\freq_v} = (A^{\freq_v}\cup B^{\freq_v}, I^{\freq_v})$ is a bipartite expander such that $|A^{\freq_v}|\geq n$, we can interpret each vertex $u_i$ as a vertex in $A^{\freq_v}$. Then, since $k<2^{r+1}$, plugging in $K =k$ in using \Cref{multi}, $H^{\freq_v}$ is always a $(D, k)$ expander where
	$$\begin{aligned}
		D &\geq q - \frac{\log_q n\cdot (b_0+2)}{2}\cdot (qk)^{1/(b_0+2)}\\
		&\geq q - 0.5\cdot\lambda^{1 / (b_0+2)} \cdot q\cdot \lambda^{1/(b_0+2)} = q/2 > 1
	\end{aligned}$$
	Hence, we can find a perfect matching $\{(u_i, r_i)\mid 1\leq i\leq k\}$ in $H^{\freq_v}$, where $1\leq r_i\leq q^{b_0+2}$; recall that $\lambda = \ceil{\log^{2+3/\delta}n\cdot (2 + 3/\delta)^{2+3/\delta}}$ and $b_0 = \ceil{\log_q(\lambda\cdot 2^{r+1})} < 3/\delta$ by \Cref{param-ineq}. Additionally, assume for each $1\leq i\leq k$, edge $(u_i, r_i)$ is the $t_i$-th edge in the adjacency list of $u_i$ in graph $H^{{\freq_v}}$; note that $1\leq t_i\leq q$ by construction of multiplicity codes.
	
	Order all the edges in $F_{l, r}$ alphabetically. For each edge $(u_i, v)\in F_{l, r}$, assume it is the $j$-th edge around $u_i$. To encode the colors in $\clr_v$, we interpret each color as a tuple $\in [q]\times [q^{b_0+2}]\times [2^{l+1}]$. Then, assign the $\kappa_i$-th color in $\clr_v$ to $(u_i, v)$ and print it in the output stream, where we set:
	$$\kappa_i = \brac{t_i, \cnt(\clr_v) + r_i, j}$$
	As a minor issue, when $\cnt(\clr_v) + r_i$ is larger than $q^{b_0+2}$, we take its modulo and replace it with $\cnt(\clr_v) + r_i \mod q^{b_0+2}$.
\end{itemize}

\paragraph*{Postprocessing Marked Sets.} After processing the input batch $F$, for every vertex $v\in R$ such that $\deg_F(v)\in [2^r, 2^{r+1})$, add node $F$ to $M_{v, \freq_v}$ with $\overrightarrow{c_F} = (c_1, c_2, \ldots, c_h)$. The whole algorithm is summarized as \Cref{det-alg-low}.

\begin{algorithm}
    \SetNoFillComment
    \caption{$\textsc{DetColorLowDeg}(F)$}\label{det-alg-low}
        \tcc{pre-processing marked sets}
    \For{$v\in R$ and frequency vector $\freq$}{
        call $\textsc{DetUpdateMarkSet}(F)$ (as described in \Cref{alg:detupdatemarkset}) to remove marked nodes in previous tree components in $\tree_\freq$,\\
        and elevate the positions of all the marked nodes in $M_{v, \freq}$ to their ancestors which are children of $P$,\\
        and assign the tuples $\overrightarrow{c_*}$ properly for newly marked nodes\;
    }
    \tcc{select frequency vector $\freq_v$ and palette $\clr_v$}
    \For{$v\in R$ such that $\deg_F(v)\in [2^r, 2^{r+1})$}{
        call $\textsc{DetFreqVec}(F)$ (as described in \Cref{alg:detfreqvec}) to grow a vector $\freq_q = (f_1, f_2, \ldots )$ progressively\;
        select palette $\clr_v = (c_1, c_2, \ldots, c_h)$ using online perfect matching on bipartite expanders $H_0^{\freq_v}, H_1^{\freq_v}, \ldots, H_{h-1}^{\freq_v}$\;
    }
    \tcc{assign colors for $F_{l, r}$}
    \For{$v\in R$ such that $\deg_F(v)\in [2^r, 2^{r+1})$}{
        let $u_1, u_2, \ldots, u_k$ be all neighbors of $v$ in $F_{l, r}$\;
        find a perfect matching $\{(u_i, r_i)\mid 1\leq i\leq k\}$ in $H^{\freq_v}$\;\label{perfect-matching-low}
        assign the $\kappa_i$-th color in $\clr_v$ to $(u_i, v)$, with $\kappa_i = (t_i, \cnt(\clr_v) + r_i, j)$, assuming $(u_i, r_i)$ is the $t_i$-th edge of $u_i$ in $H^{\freq_v}$, and $(u_i, v)$ is the $j$-th edge around $u_i$ in $F_{l, r}$\;
    }
    \tcc{post-processing marked sets}
    \For{$v\in R$ such that $\deg_F(v)\in [2^r, 2^{r+1})$}{
        add $F$ to $M_{v, \freq_v}$ with $\overrightarrow{c_F} = (c_1, c_2, \ldots, c_h) = \clr_v$\;
    }
\end{algorithm}

\subsubsection{Proof of Correctness}
To begin with, let us first bound the total number of colors that we use throughout the algorithm.
\begin{lemma}
	The total number of colors used by the algorithm is at most $$O\brac{(\log\Delta)^{O(1 / \epsilon)}\cdot(1/\epsilon)^{O(1/\epsilon^2)}\cdot \Delta^{1+2\epsilon}\cdot 2^l}$$ different colors.
\end{lemma}
\begin{proof}
	It suffices to estimate the number of colors we have created for each frequency vector $\freq = (f_1, f_2, \ldots, f_h)$. By the color allocation process, the total number of colors in $\clr^\freq$ is bounded equal to (for a single tree component in $\tree_{\freq}$):
	$$\begin{aligned}
		|\clr^\freq| &= 2^{l+1}\cdot q^{b_0+3}\cdot\prod_{i=1}^h(\lambda_i\cdot q^{b_i+2})\\
		&= \lambda^{h}\cdot 2^{l+1}\cdot q^{2h+3}\cdot\prod_{i=0}^h q^{b_i}\cdot \prod_{i=1}^h \ceil{(b_i+2)\cdot \log_2q}\\
		&\leq \lambda^{h}\cdot 2^{l+1}\cdot q^{2h+3}\cdot 2^{r+1}\cdot\prod_{i=1}^h (\lambda\cdot f_i\cdot q)\cdot \prod_{i=1}^h \ceil{(b_i+2)\cdot \log_2q}\\
		&\leq \lambda^{2h}\cdot 2^{l+1}\cdot q^{3h+3}\cdot \prod_{i=1}^h \ceil{(b_i+2)\cdot \log_2q} \cdot \prod_{i=0}^h f_i
	\end{aligned}$$
	Since there are at most $\ceil{\Delta / \prod_{i=0}^{h - 1}f_i}$ different tree components, the total number of colors would be:
	$$\lambda^{2h}\cdot 2^{l+1}\cdot q^{3h+3}\cdot \prod_{i=1}^h \ceil{(b_i+2)\cdot \log_2q} \cdot 2\Delta \cdot f_h$$
	
	Note that $\lambda = \ceil{\log^{2+3/\delta}n\cdot (2 + 3/\delta)^{2+3/\delta}} < (4/\delta)^{4/\delta} \cdot \log^{4/\delta} n < (4/\delta)^{4/\delta}\cdot \Delta^{\epsilon^2 / 5}$, and $h\leq 1/\epsilon$. Also, for each $1\leq i\leq h$, $b_i = \ceil{\log_q(\lambda\cdot f_i)} < \ceil{\log_q \Delta^2} < 1 + 2/\delta < 3/\delta$. Therefore, in the end, the total number of colors for each $\freq$ is (recall that $\delta = \epsilon^2/10$, $q < 2\Delta^{\delta}$):
	$$\begin{aligned}
		&\lambda^{2h}\cdot 2^{l+1}\cdot q^{3h+3}\cdot \prod_{i=1}^h \ceil{(b_i+2)\cdot \log_2q}\cdot 2\Delta\cdot f_h\\
		&\leq (4/\delta)^{8/\epsilon \delta}\cdot \Delta^{0.4\epsilon}\cdot 2^{l+1}\cdot (2\Delta)^{\delta \cdot (3 + 3/\epsilon)} \cdot \ceil{(2+3/\delta)\cdot \log_2 q}^{1/\epsilon}\Delta\cdot 2\Delta \cdot \Delta^{\epsilon}\\
		&= (\log\Delta)^{O(1 / \epsilon)}\cdot(1/\epsilon)^{O(1/\epsilon^3)}\cdot \Delta^{1+2\epsilon}\cdot 2^l
	\end{aligned}$$
\end{proof}

\begin{lemma}\label{palette-cnt}
	Each palette is used in at most $2^{r+1}$ different batches.
\end{lemma}
\begin{proof}
	Consider any palette specified by a tuple $(c_1, c_2, \ldots, c_h)$. Recall that $f_0 = 2^{r+1}$, so it suffices to prove by an induction that the number of level-$k$ tree nodes in the same connected component in $\tree_\freq$ with the same color package $(*, *, \ldots, *, c_{k+1}, c_{k+2}, \ldots, c_h)$ is at most $f_k$, for any $0\leq k\leq h$. The basis is trivial because the level-$h$ tree node in a component in $\tree_{\freq}$ is its root.
	
	By the allocation of color packages, for any node $N$ on level-$(k-1)$ with package specified by tuple $(*, *, \ldots, *, c_k, c_{k+1}, \ldots, c_h)$, its parent has color package $(*, *, \ldots, *, c_{k+1}, c_{k+2}, \ldots, c_h)$. So, it suffices to bound the number of children with color package $(*, *, \ldots, *, c_k, c_{k+1}, \ldots, c_h)$ for each level-$k$ node with color package $(*, *, \ldots, *, c_{k+1}, c_{k+2}, \ldots, c_h)$. Take any such level-$k$ node $W$ and its $s$-th child $N$, $N$ has color package $(*, *, \ldots, *, c_k, c_{k+1}, \ldots, c_h)$ if and only if $(s, c_k)$ is an edge in graph $H_{k-1}^\freq$. Since $H_{k-1}^\freq$ is a union of copies of the bipartite expanders based on multiplicity codes (to be more precise, $H_{k-1}^\freq$ is a union of several bipartite expanders with the same left-side and separate right-sides), by \Cref{deg}, the number of such children $N$ is at most
	$$q^{\ceil{\log_q f_{k-1}} - b_k - 1} \leq f_{k-1} / q^{b_k} \leq f_{k-1} / f_k$$
	which concludes the induction.
\end{proof}

\begin{corollary}
	During the algorithm, the values of the counters $\cnt(\clr)$ never exceed $2^{r+1}$ for any palette $\clr$.
\end{corollary}

As before, we can analyze the total space of the marked sets under the condition of \Cref{det-inv}(2). Since the proof would be the same, we omit it here.

\begin{lemma}
	If \Cref{det-inv} is satisfied, then the total size $\sum_{v\in R}|M_{v, \freq}|$ is bounded by $O(n/\epsilon)$ for any frequency vector $\freq$.
\end{lemma}

Now, let us verify that our algorithm preserves \Cref{det-inv}.
\begin{lemma}
	\Cref{det-inv} is preserved by the algorithm throughout its execution.
\end{lemma}
\begin{proof}
	\Cref{det-inv}(1)(2) are the same as \Cref{inv}(1)(2), and we can verify them following the same argument presented in \Cref{verify-inv}, so let us focus on \Cref{det-inv}(3)(4) for the rest.
	
	For \Cref{det-inv}(3), according to the pre-processing step upon the arrival of each input batch, a new marked node on level-$i$ would inherit at least the last $h-i$ coordinates from its marked descendants. So this property should hold.
	
	Let us now verify \Cref{det-inv}(4). Consider any pair of marked nodes $N_1, N_2$ and let $W$ be their lowest ancestor which is on level-$k$. Assume $N_2$ is chronologically later than $N_1$. If $N_2$ is a leaf node, then according to the algorithm, $\overrightarrow{c_{N_2}}$ would share the last $h-k$ coordinates with $W$. Otherwise, for the same technical reason with \Cref{det-inv}, according to the pre-processing step of each iteration, $\overrightarrow{c_{N_2}}$ would inherit at least the last $h-k$ coordinates from its marked descendants, so this property still holds.
	
	For the second half of the statement of \Cref{det-inv}(4), assume $N_1, N_2, \ldots, N_i$ are the marked children of node $W$ on level-$k$, and suppose $N_j$ is the $k_j$-th child of $W$ with $k_1< k_2<\cdots < k_i$. According to the algorithm, it selected the pair $\brac{k_1, c_k^{N_1}}$ by initializing an online perfect matching algorithm from \Cref{perfect} in bipartite graph $H_{k-1}^\freq$, and view $N_1$ as the first online vertex matched using edge $\brac{k_1, c_k^{N_1}}$. For $1<j\leq i$, the online matching algorithm continues to find an edge $\brac{k_j, c_k^{N_j}}\in E(H_{k-1}^\freq) = I_{k-1}^\freq$ which is different from previous matching edges. By \Cref{verify-perfect}, the online matching algorithm could always find a matching edge $\brac{k_j, c_k^{N_j}}$ whenever $j\leq i\leq f_{k-1}$.
\end{proof}

Finally, let us verify the validity of the output.
\begin{lemma}
	In the output stream, the algorithm never prints the same color for two adjacent edges.
\end{lemma}
\begin{proof}
	First, let us rule out color conflicts for edges around the same vertex $u\in L$. On the one hand, according to the algorithm description, within each batch, we only assign tentative colors with distinct color indices around each vertex $u\in L$. On the other hand, for two different batches $F, F'$, if we happen to use the same color palette $\clr$ in $F_{l, r}, F_{l, r}'$ around the same vertex $u$. Consider any edge $(u, v_1)\in F_{l, r}$ and edge $(u, v_2)\in F_{l, r}'$, and assume $(u, v_1)$ receives the $\kappa$-th color and $(u, v_1)$ receives the $\kappa'$-th color in palette $\clr$, $\kappa$ and $\kappa'$ taking the following form:
	$$\kappa = \brac{t, \cnt_1 + r, j}$$
	$$\kappa' = \brac{t', \cnt_2 + r', j'}$$
	To explain the parameters, according the algorithm, $\cnt_1$ and $\cnt_2$ denote the values of $\cnt(\clr_v)$ at the arrivals of batch $F$ and $F'$, respectively; $(u, r), (u, r')\in I^\freq = E(H^\freq)$ are the matching edges in $H^\freq$ we find using \Cref{multi}, and they are the $t$-th and $t'$-th edge around $u$ in $H^\freq$; finally, $(u, v_1)$ and $(u, v_2)$ are the $j$-th and the $j'$-the edge around $u$ in batch $F_{l, r}, F'_{l, r}$. To argue that $\kappa\neq\kappa'$, if $t = t'$, then we must have $r = r'$ by its definition. Since in batches $F, F'$ we use the same palette $\clr_v$ and by \Cref{palette-cnt}, we must have $\cnt_1\neq \cnt_2 \mod q^{b_0+2}$, so $\kappa\neq \kappa'$.
	
	Next, let us rule out color conflicts for edges around the same vertex $v\in R$. For any color palette provided by $\tree_{\freq_v}$ we have used before in some previous batch $F'$, according to \Cref{det-inv}(3), there must be a marked node $N$ whose subtree contains $F'$. According to our coloring procedure, $\clr_v$ is nonempty only when $\clr^{N}$ is disjoint from all the color packages of its marked siblings, for any node $N$ on the root-to-leaf path to $F$. Therefore, $v$ could not reuse any previously assigned colors.
\end{proof}

\subsection{Proof of \Cref{det-high}}
\subsubsection{Basic Notations}
Define a constant parameter $\delta = \epsilon^2/10$. As before, we can assume $\Delta$ is a power of $2$ without loss of generality. If $\Delta < \log^{10/\epsilon\delta}n$, then we will store the entire graph and output a $\Delta$-edge coloring. For the rest, let us assume $\Delta > \log^{10/\epsilon\delta}n$.

Throughout the algorithm, each vertex $u\in L$ maintains a value $\cnt(u)$ that counts the number of input batches $F$ where $\deg_F(u)\in [2^l, 2^{l+1})$; similarly, each vertex $v\in R$ maintains a value $\cnt(v)$ that counts the number of input batches $F$ where $\deg_F(v)\in [2^r, 2^{r+1})$. 

Take a prime number $q\in [\Delta^\delta, 2\Delta^\delta)$. Set parameters:
$$\lambda = \ceil{\log^{2+3/\delta}_2 n\cdot (2 + 3/\delta)^{2+3/\delta}}$$
$$a_1 = \ceil{\log_q (n\Delta / 2^l)}+2, b_1 = \ceil{\log_q (\lambda\Delta / 2^l)}$$
$$a_2 = \ceil{\log_q (n\Delta / 2^r)}+2, b_2 = \ceil{\log_q (\lambda\Delta / 2^r)}$$
Apply \Cref{multi} and construct two bipartite expanders based on multiplicity codes $H_1 = (A_1\cup B_1, I_1), H_2 = (A_2\cup B_2, I_2)$, where $A_i = \field_q^{a_i}, B_i = \field_q^{b_i+2}$. 

\begin{lemma}
    For $i\in \{1,2\}$, we have $b_i< 2/\delta$, and $\lambda \geq a_i^{b_i+2}\cdot (b_i+2)^{b_i+2}$.
\end{lemma}
\begin{proof}
    By the assumption that $\Delta > \log^{10/\epsilon\delta}n$ and $\delta = \epsilon^2/10$, the following holds when $n$ is super-constant:
    $$\lambda = \ceil{\log_2^{2+3/\delta}n\cdot (2 + 3/\delta)^{2+3/\delta}} <\log^{\frac{10}{3\delta}}n< \Delta^{\epsilon/3}$$
    Hence, $b_i \leq \ceil{\log_q \Delta^{1+\epsilon/3}} < 2/\delta$ as $q > \Delta^\delta$. Furthermore, as $a_i \leq \log_q (n\Delta) + 3 < 2\log_q n$, we have:
    $$\begin{aligned}
        \lambda &\geq \log_2^{2+3/\delta}n\cdot (2+3/\delta)^{2+3/\delta}\\
        &> (a_i/2)^{2+3/\delta}\cdot (b_i+2)^{b_i+2}\\
        &> a_i^{2+2/\delta}\cdot (b_i+2)^{b_i+2} \geq a_i^{b_i+2}\cdot (b_i+2)^{b_i+2}
    \end{aligned}$$
    The second last inequality holds as $(a_i/2)^{2+3/\delta} > a_i^{2+2/\delta}$ for super-constant choices of $n$.
\end{proof}

\begin{lemma}
If we have $O\brac{\Delta^{1+\epsilon} + \Delta^{2+\epsilon} / 2^{l+r}}$ colors, then we can encode each color as a tuple from $[q^{b_1+2}]\times [q^{b_2+2}]\times [\ceil{2^{l+r+2} / \Delta}]$.
\end{lemma}
\begin{proof}
    As before, we can show $\lambda < \Delta^{\epsilon/3}$. By the setup of the parameters, the total number of such tuples is bounded by:
    $$\begin{aligned}
        q^{b_1+b_2+4}\cdot \brac{1 + \frac{2^{l+r+2}}{\Delta}}&\leq q^4\cdot \lambda^2\cdot \frac{\Delta^2}{2^{l+r}}\cdot \brac{1 + \frac{2^{l+r+2}}{\Delta}}\\
        &\leq \frac{\lambda^2\cdot\Delta^{2+4\delta}}{2^{l+r}} + 4\Delta^{1+4\delta}\cdot \lambda^2\\
        &\leq \frac{\Delta^{2 + \frac{2\epsilon}{3} + \frac{2\epsilon^2}{5}}}{2^{l+r}} + 4\Delta^{1 +  \frac{2\epsilon}{3} + \frac{2\epsilon^2}{5}}\\
        &= O(\Delta^{1+\epsilon} + \Delta^{2+\epsilon} / 2^{l+r})
    \end{aligned}$$
    which concludes the calculation.
\end{proof}

At any moment in time during the algorithm, for each $u\in L$, since $0\leq \cnt(u) < \Delta / 2^l$ and $|L| < n$, we can interpret $\cnt(u)$ as a polynomial $g_u(X)\in \field_q^{<\ceil{\log_q \Delta / 2^l}}[X]$, and vertex $u$ as a polynomial $f_u(X)\in \field_q^{<\ceil{\log_q n}}[X]$. Define polynomial $h_u(X) = g_u(X) + X^{\ceil{\log_q \Delta / 2^l}}\cdot f_u(X)$.

Symmetrically, for each $v\in R$, we can interpret $\cnt(v)$ as a polynomial $g_v(X)\in \field_q^{\ceil{\log_q \Delta / 2^r}}[X]$, and vertex $v$ as a polynomial $f_v(X)\in \field_q^{<\ceil{\log_q n}}[X]$. Define polynomial $h_v(X) = g_v(X) + X^{\ceil{\log_q \Delta / 2^r}}\cdot f_v(X)$. By definition, $f_v$ does not change over time, and $h_v, g_v$ are dependent on how many batches we have read so far.

\subsubsection{Algorithm Description}
Upon the arrival of an input batch $F$, let us discuss how to assign colors to edges in $F_{l, r}$. We are going to associate a pair $(s_e, t_e)\in [q^{b_1+2}]\times [q^{b_2+2}]$ for every edge $e\in F_{l, r}$ in the following way. Go over every vertex $u\in L$ and list all of its edges $(u, v_1), (u, v_2), \ldots, (u, v_k)$ in $F_{l, r}$. Since $G$ is a simple graph, all the neighbors $v_1, v_2, \ldots, v_k$ are different. 

Divide $\{v_1, v_2, \ldots, v_k\}$ into $\ceil{2^{l+r+2} / \Delta}$ groups, each of size at most $\Delta / 2^r$. For each of these group, say that it contains vertices $z_1, z_2, \ldots, z_K, K\leq \Delta / 2^r$. Interpret each polynomial $h_{z_i}$ as a vertex in $A_2$. According to \Cref{multi}, $H_2$ is a $(D, K)$-expander where:
$$\begin{aligned}
	D = q - \frac{a_2(b_2+2)}{2}\cdot (qK)^{1 / (b_2+2)}
	\geq q - 0.5\cdot \lambda^{1 / (b_2+2)}\cdot q\cdot \lambda^{1 / (b_2+2)} = q/2 >1
\end{aligned}$$
Therefore, there exists a bipartite matching $\{(h_{z_i}, t_i)\mid 1\leq i\leq K\}$ of $H_2$. Since $t_i\in \field_q^{b_2+2}$, we can interpret $t_i$ as an integer $t_{(u, z_i)}\in [q^{b_2+2}]$.

Symmetrically, for each edge $e = (u, v)$ we can also define the value $s_e$ for every edge $e\in F_{l, r}$. Next, for any integer pair $(s, t)\in [q^{b_1+2}]\times [q^{b_2+2}]$, define $F_{l, r}^{(s, t)} = \{e\in F_{l, r}\mid (s_e, t_e) = (s, t)\}$.
\begin{lemma}
	For any $(s,  t)$, the maximum degree of edge set $F_{l, r}^{(s, t)}$ is at most $\ceil{2^{l+r+2} / \Delta}$.
\end{lemma}
\begin{proof}
	Consider any vertex $u\in L$. According to the definition of $t_e$'s, in each group of vertices $\{z_1, z_2, \ldots, z_K\}$, there is at most one $z_i$ such that $t_{(u, z_i)} = t$. Since all neighbors of $u$ are grouped into at most $\ceil{2^{l+r+2} / \Delta}$ groups, the degree of $u$ in $F_{l, r}^{(s, t)}$ is bounded by $\ceil{2^{l+r+2} / \Delta}$. Symmetrically, we can show the same upper bound on degrees on $R$.
\end{proof}

By the above statement, we can use the set of colors $\{(s, t, c)\mid c\in [\ceil{2^{l+r+2} / \Delta}]\}$ to color all edges in $F_{l, r}^{(s, t)}$. After the coloring is completed for all pairs $(s, t)$, we print all the colors of $F_{l, r}$ in the output stream. The whole algorithm is summarized in \Cref{det-alg-high}.

\begin{algorithm}
    \caption{$\textsc{DetColorHighDeg}(F)$}\label{det-alg-high}
    compute polynomials $h_u(X), \forall u\in V$ according to $\cnt(u)$\;
    define $E_{x, y} = \{(u, v)\in F_{l, r}\mid (x_u, y_u) = (x, y)\}, \forall (x, y)$\;
    \For{$u\in L$ such that $\deg_F(u)\in [2^l, 2^{l+1})$}{
        let $v_1, v_2, \ldots, v_k$ be all of $u$'s neighbor in $F_{l, r}$\;\label{distinct-neighbors}
        compute integers $t_{(u, v_i)}, 1\leq i\leq k$ using bipartite matchings\;\label{perfect-matching-high}
    }
    compute integers $s_{e}, \forall e\in F_{l, r}$ by a symmetric manner\;
    compute $F_{l, r}^{(s, t)} = \{e\in F_{l, r}\mid (s_e, t_e) = (s, t)\}$ for all $(s, t)\in [q^{b_1+2}]\times [q^{b_2+2}]$ and color each $F_{l, r}^{(s, t)}$ using $\ceil{2^{l+r+2}/\Delta}$ colors\;
    increment counters $\cnt(*)$\;
\end{algorithm}

\subsubsection{Proof of Correctness}
It suffices to show that the algorithm always produces a valid edge coloring of $G$.
\begin{lemma}
	In the output stream, the algorithm never prints the same color for two adjacent edges.
\end{lemma}
\begin{proof}
	It is clear that we never assign the same color twice around the same vertex in a single batch $F_{l, r}$. So it is enough to verify that there are no color conflicts across different input batches.
	
	Consider any $u\in L$ and two different input batches $F, F'$ such that $\deg_F(u), \deg_{F'}\in [2^l, 2^{l+1})$, and any edge $(u, v)\in F_{l, r}$ and edge $(u, v')\in F_{l, r}'$. Assume $(u, v), (u, v')$ have colors $(s, t, k), (s', t', k')\in [q^{b_1}]\times [q^{b_2}]\times [\ceil{2^{l+r+2} / \Delta}]$ in the output stream, respectively.
	
	Let $h(X), h'(X)$ refer to the polynomial of $h_u(X)$ at the arrival of $F$ and $F'$, respectively. Similarly, let us define $g(X), g'(X)$ be the polynomial of $g_u(X)$ at the arrival of $F, F'$. Therefore, $h(X) = g(X) + X^{\ceil{\log_q\Delta / 2^l}}\cdot f_u(X)$, and $h'(X) = g'(X) + X^{\ceil{\log_q\Delta / 2^l}}\cdot f_u(X)$.
	
	Assume $s = (x, x_0, x_1, \ldots, x_{b_1})$, and $s' = (x', x'_0, x'_1, \ldots, x'_{b_1})$. By the choice of the value $s$, we know that $\Gamma(h, x) = s, \Gamma(h', x') = s'$. If $s = s'$, then $x = x'$, and by linearity of derivatives, we have $g^{(i)}(x) = g'^{(i)}(x), \forall 0\leq i\leq b_1$. Since both $g, g'$ have degree less than $\ceil{\log_q \Delta/2^l} \leq b_1$, it must be the case that $g \equiv g'$. However, this is impossible because the values of the counter $\cnt(u)$ must be different for batch $F, F'$.
	
	Symmetrically, we can show that there are no color conflicts around vertices in $R$. This concludes the proof.
\end{proof}

\subsubsection{Extension to Multi-graphs}
So far we have only considered simple graphs. Let us briefly discuss how to extend our algorithm to multi-graphs. Basically, in multi-graphs it could be the case that $\Delta \gg n$, and so storing the forest structures and bipartite expanders would be too costly. Technically speaking, There are mainly three places where we needed the input graph to be simple.
\begin{enumerate}[(1)]
    \item The size of the forest structure is at least $\Delta$; 
    \item Finding perfect matchings on \Cref{perfect-matching-low} in \Cref{det-alg-low} and on \Cref{perfect-matching-high} in \Cref{det-alg-high} requires explicitly storing some bipartite expanders of size larger than $\Delta$; 
    \item On \Cref{distinct-neighbors} of \Cref{det-alg-high} we need all neighbors of $u$ to be distinct. 
\end{enumerate}
Let us discuss how to bypass these three issues.

For issue (1), the main observation is that we can store the entire forest structures implicitly, since the color package allocations are defined by multiplicity codes in a closed-form; note that we could do this with the randomized algorithm, since we had to remember all the randomness in order to store the color package allocations.

For issue (2), instead of performing a standard perfect matching algorithm on the bipartite expander, we could apply the greedy matching algorithm on the bipartite expander from \Cref{perfect}, which only takes space proportional to the matching size.

For issue (3), note that we only require that the subgraph $(V, F)$ is simple, not necessarily the whole input graph $G$. Therefore, it suffices to use a reduction from a general input stream to input streams of batches which are simple subgraphs, which is presented below.

\paragraph*{Multi- to Simple Reduction.} For the above technical reason, we need the assumption that every input batch is a simple subgraph. We argue that this extra condition does not make the problem simpler.
\begin{lemma}
	Given an algorithm $\mathcal{A}$ for coloring graphs using $f(\Delta)$ colors and $g(n, \Delta)$ space under the condition that every input batch is a simple subgraph of $G$, there is an algorithm $\mathcal{B}$ using $O(f(\Delta)\log\Delta)$ colors and $O(g(n, \Delta)\log\Delta)$ space for coloring any graph streams. Furthermore, this reduction is deterministic.
\end{lemma}
\begin{proof}
	We are proving the statement by an induction on $\Delta$. The basis is trivial for constant $\Delta$. Assume this holds for smaller values of $\Delta$. For the reduction, the idea is to extract input batches of simple subgraphs and feed them to $\mathcal{A}$, and take the rest of the input stream to an algorithm $\mathcal{B}'$ that handles graphs with maximum degree $\floor{\Delta/2}$.
	
	More specifically, in each round we prepare an edge set $F$ of a simple subgraph of $G$, and when $|F|$ reaches $n$, we take it to algorithm $\mathcal{A}$, and start the next round. To grow the set $F$, we read the next edge $e$ from the input stream. If $F\cup \{e\}$ is a simple subgraph, then add $e$ to $F$. Otherwise, take an edge $e'\in F$ which is parallel to $e$, and remove $e'$ from $F$. At the same time, send the edge $e$ to algorithm $\mathcal{B}'$ for coloring any graph stream of degree $\floor{\Delta/2}$. When $\mathcal{B}'$ is about to print a color $k$ for edge $e$, we assign color $2k-1$ to edge $e$ and color $2k$ to edge $e'$ in the real output stream.
	
	In total, we are using $f(\Delta) + 2\cdot O(f(\floor{\Delta/2}))$ colors and $g(n, \Delta) + O(g(n, \floor{\Delta/2}))$ space in total. Solving the recurrence concludes the proof.
\end{proof}

\section*{Acknowledgment}
Shiri Chehcik is funded by the European Research Council (ERC) under the European Union’s Horizon 2020 research and innovation programme (grant agreement No 803118 UncertainENV). Tianyi Zhang is financially supported by the starting grant ``A New Paradigm for Flow and Cut Algorithms'' (no. TMSGI2\_218022) of the Swiss National Science Foundation.

\vspace{5mm}
\bibliographystyle{alpha}
\bibliography{references}

\newcommand{\etalchar}[1]{$^{#1}$}
\begin{thebibliography}{GKMU18}

\bibitem[ABB{\etalchar{+}}24]{assadi2024vizing}
Sepehr Assadi, Soheil Behnezhad, Sayan Bhattacharya, Mart{\'\i}n Costa, Shay Solomon, and Tianyi Zhang.
\newblock Vizing's theorem in near-linear time.
\newblock {\em arXiv preprint arXiv:2410.05240}, 2024.

\bibitem[Arj82]{arjomandi1982efficient}
Eshrat Arjomandi.
\newblock {An efficient algorithm for colouring the edges of a graph with $\Delta+1$ colours}.
\newblock {\em INFOR: Information Systems and Operational Research}, 20(2):82--101, 1982.

\bibitem[Ass25]{Assadi24}
Sepehr Assadi.
\newblock {Faster Vizing and Near-Vizing Edge Coloring Algorithms}.
\newblock In {\em Annual ACM-SIAM Symposium on Discrete Algorithms (SODA)}, 2025.

\bibitem[ASZZ22]{ansari2022simple}
Mohammad Ansari, Mohammad Saneian, and Hamid Zarrabi-Zadeh.
\newblock Simple streaming algorithms for edge coloring.
\newblock In {\em 30th Annual European Symposium on Algorithms (ESA 2022)}. Schloss-Dagstuhl-Leibniz Zentrum f{\"u}r Informatik, 2022.

\bibitem[BBKO22]{balliu2022distributed}
Alkida Balliu, Sebastian Brandt, Fabian Kuhn, and Dennis Olivetti.
\newblock {Distributed edge coloring in time polylogarithmic in $\Delta$}.
\newblock In {\em Proceedings of the 2022 ACM Symposium on Principles of Distributed Computing}, pages 15--25, 2022.

\bibitem[BCC{\etalchar{+}}24]{BhattacharyaCCSZ24}
Sayan Bhattacharya, Din Carmon, Mart{\'\i}n Costa, Shay Solomon, and Tianyi Zhang.
\newblock {Faster $(\Delta+1)$-Edge Coloring: Breaking the $m\sqrt{n}$ Time Barrier}.
\newblock In {\em 65th IEEE Symposium on Foundations of Computer Science (FOCS)}, 2024.

\bibitem[BCHN18]{BhattacharyaCHN18}
Sayan Bhattacharya, Deeparnab Chakrabarty, Monika Henzinger, and Danupon Nanongkai.
\newblock {Dynamic Algorithms for Graph Coloring}.
\newblock In {\em Proceedings of the Twenty-Ninth Annual {ACM-SIAM} Symposium on Discrete Algorithms (SODA)}, pages 1--20. {SIAM}, 2018.

\bibitem[BCPS24]{BhattacharyaCPS24}
Sayan Bhattacharya, Mart\'in Costa, Nadav Panski, and Shay Solomon.
\newblock {Nibbling at Long Cycles: Dynamic (and Static) Edge Coloring in Optimal Time}.
\newblock In {\em Proceedings of the {ACM-SIAM} Symposium on Discrete Algorithms (SODA)}. {SIAM}, 2024.

\bibitem[BCSZ25]{BhattacharyaCSZ24}
Sayan Bhattacharya, Mart{\'\i}n Costa, Shay Solomon, and Tianyi Zhang.
\newblock {Even Faster $(\Delta+1)$-Edge Coloring via Shorter Multi-Step Vizing Chains}.
\newblock In {\em Annual ACM-SIAM Symposium on Discrete Algorithms (SODA)}, 2025.
\newblock (to appear).

\bibitem[BDH{\etalchar{+}}19]{behnezhad2019streaming}
Soheil Behnezhad, Mahsa Derakhshan, MohammadTaghi Hajiaghayi, Marina Knittel, and Hamed Saleh.
\newblock Streaming and massively parallel algorithms for edge coloring.
\newblock In {\em 27th Annual European Symposium on Algorithms (ESA 2019)}. Schloss Dagstuhl-Leibniz-Zentrum fuer Informatik, 2019.

\bibitem[Ber22]{Bernshteyn22}
Anton Bernshteyn.
\newblock {A fast distributed algorithm for $(\Delta+1)$-edge-coloring}.
\newblock {\em J. Comb. Theory, Ser. {B}}, 152:319--352, 2022.

\bibitem[BGW21]{BhattacharyaGW21}
Sayan Bhattacharya, Fabrizio Grandoni, and David Wajc.
\newblock Online edge coloring algorithms via the nibble method.
\newblock In {\em Proceedings of the{ACM-SIAM} Symposium on Discrete Algorithms (SODA)}, pages 2830--2842. {SIAM}, 2021.

\bibitem[BM17]{BarenboimM17}
Leonid Barenboim and Tzalik Maimon.
\newblock Fully-dynamic graph algorithms with sublinear time inspired by distributed computing.
\newblock In {\em International Conference on Computational Science (ICCS)}, volume 108 of {\em Procedia Computer Science}, pages 89--98. Elsevier, 2017.

\bibitem[BSVW24]{BilkstadSVW24}
Joakim Blikstad, Ola Svensson, Radu Vintan, and David Wajc.
\newblock Online edge coloring is (nearly) as easy as offline.
\newblock In {\em Proceedings of the Annual {ACM} Symposium on Theory of Computing (STOC)}. {ACM}, 2024.

\bibitem[BSVW25]{BlikstadOnline2025}
Joakim Blikstad, Ola Svensson, Radu Vintan, and David Wajc.
\newblock {Deterministic Online Bipartite Edge Coloring}.
\newblock In {\em Annual ACM-SIAM Symposium on Discrete Algorithms (SODA)}, 2025.

\bibitem[CHL{\etalchar{+}}20]{ChangHLPU20}
Yi{-}Jun Chang, Qizheng He, Wenzheng Li, Seth Pettie, and Jara Uitto.
\newblock {Distributed Edge Coloring and a Special Case of the Constructive Lov{\'{a}}sz Local Lemma}.
\newblock {\em {ACM} Trans. Algorithms}, 16(1):8:1--8:51, 2020.

\bibitem[Chr23]{Christiansen23}
Aleksander Bj{\o}rn~Grodt Christiansen.
\newblock {The Power of Multi-step Vizing Chains}.
\newblock In {\em Proceedings of the 55th Annual {ACM} Symposium on Theory of Computing (STOC)}, pages 1013--1026. {ACM}, 2023.

\bibitem[Chr24]{Christiansen24}
Aleksander B.~G. Christiansen.
\newblock Deterministic dynamic edge-colouring.
\newblock {\em CoRR}, abs/2402.13139, 2024.

\bibitem[CL21]{charikar2021improved}
Moses Charikar and Paul Liu.
\newblock Improved algorithms for edge colouring in the w-streaming model.
\newblock In {\em Symposium on Simplicity in Algorithms (SOSA)}, pages 181--183. SIAM, 2021.

\bibitem[CMZ24]{chechik_et_al:LIPIcs.ICALP.2024.40}
Shiri Chechik, Doron Mukhtar, and Tianyi Zhang.
\newblock {Streaming Edge Coloring with Subquadratic Palette Size}.
\newblock In Karl Bringmann, Martin Grohe, Gabriele Puppis, and Ola Svensson, editors, {\em 51st International Colloquium on Automata, Languages, and Programming (ICALP 2024)}, volume 297 of {\em Leibniz International Proceedings in Informatics (LIPIcs)}, pages 40:1--40:12, Dagstuhl, Germany, 2024. Schloss Dagstuhl -- Leibniz-Zentrum f{\"u}r Informatik.

\bibitem[COS01]{cole2001edge}
Richard Cole, Kirstin Ost, and Stefan Schirra.
\newblock Edge-coloring bipartite multigraphs in o (e logd) time.
\newblock {\em Combinatorica}, 21(1):5--12, 2001.

\bibitem[CPW19]{CohenPW19}
Ilan~Reuven Cohen, Binghui Peng, and David Wajc.
\newblock Tight bounds for online edge coloring.
\newblock In {\em 60th {IEEE} Annual Symposium on Foundations of Computer Science (FOCS)}, pages 1--25. {IEEE} Computer Society, 2019.

\bibitem[Dav23]{Davies23}
Peter Davies.
\newblock Improved distributed algorithms for the lov{\'{a}}sz local lemma and edge coloring.
\newblock In {\em Proceedings of the {ACM-SIAM} Symposium on Discrete Algorithms (SODA)}, pages 4273--4295. {SIAM}, 2023.

\bibitem[DFR09]{demetrescu2009trading}
Camil Demetrescu, Irene Finocchi, and Andrea Ribichini.
\newblock Trading off space for passes in graph streaming problems.
\newblock {\em ACM Transactions on Algorithms (TALG)}, 6(1):1--17, 2009.

\bibitem[DGS25]{dudeja2024randomizedgreedyonlineedge}
Aditi Dudeja, Rashmika Goswami, and Michael Saks.
\newblock {Randomized Greedy Online Edge Coloring Succeeds for Dense and Randomly-Ordered Graphs}.
\newblock In {\em Annual ACM-SIAM Symposium on Discrete Algorithms (SODA)}, 2025.

\bibitem[DHZ19]{duan2019dynamic}
Ran Duan, Haoqing He, and Tianyi Zhang.
\newblock Dynamic edge coloring with improved approximation.
\newblock In {\em 30th Annual ACM-SIAM Symposium on Discrete Algorithms (SODA)}, 2019.

\bibitem[EPS14]{elkin20142delta}
Michael Elkin, Seth Pettie, and Hsin-Hao Su.
\newblock {$(2\Delta - 1)$-Edge-Coloring is Much Easier than Maximal Matching in the Distributed Setting}.
\newblock In {\em Proceedings of the Twenty-Sixth Annual ACM-SIAM Symposium on Discrete Algorithms}, pages 355--370. SIAM, 2014.

\bibitem[FGK17]{fischer2017deterministic}
Manuela Fischer, Mohsen Ghaffari, and Fabian Kuhn.
\newblock Deterministic distributed edge-coloring via hypergraph maximal matching.
\newblock In {\em 2017 IEEE 58th Annual Symposium on Foundations of Computer Science (FOCS)}, pages 180--191. IEEE, 2017.

\bibitem[GKMU18]{ghaffari2018deterministic}
Mohsen Ghaffari, Fabian Kuhn, Yannic Maus, and Jara Uitto.
\newblock Deterministic distributed edge-coloring with fewer colors.
\newblock In {\em Proceedings of the 50th Annual ACM SIGACT Symposium on Theory of Computing}, pages 418--430, 2018.

\bibitem[GNK{\etalchar{+}}85]{gabow1985algorithms}
Harold~N Gabow, Takao Nishizeki, Oded Kariv, Daneil Leven, and Osamu Terada.
\newblock Algorithms for edge coloring.
\newblock {\em Technical Rport}, 1985.

\bibitem[GS24]{ghosh2024low}
Prantar Ghosh and Manuel Stoeckl.
\newblock Low-memory algorithms for online edge coloring.
\newblock In {\em 51st International Colloquium on Automata, Languages, and Programming (ICALP 2024)}. Schloss Dagstuhl--Leibniz-Zentrum f{\"u}r Informatik, 2024.

\bibitem[GSS17]{glazik2017finding}
Christian Glazik, Jan Schiemann, and Anand Srivastav.
\newblock Finding euler tours in one pass in the w-streaming model with o (n log (n)) ram.
\newblock {\em arXiv preprint arXiv:1710.04091}, 2017.

\bibitem[GUV09]{guruswami2009unbalanced}
Venkatesan Guruswami, Christopher Umans, and Salil Vadhan.
\newblock Unbalanced expanders and randomness extractors from parvaresh--vardy codes.
\newblock {\em Journal of the ACM (JACM)}, 56(4):1--34, 2009.

\bibitem[KLS{\etalchar{+}}22]{KulkarniLSST22}
Janardhan Kulkarni, Yang~P. Liu, Ashwin Sah, Mehtaab Sawhney, and Jakub Tarnawski.
\newblock Online edge coloring via tree recurrences and correlation decay.
\newblock In {\em 54th Annual {ACM} {SIGACT} Symposium on Theory of Computing (STOC)}, pages 104--116. {ACM}, 2022.

\bibitem[KTS22]{kalev2022unbalanced}
Itay Kalev and Amnon Ta-Shma.
\newblock Unbalanced expanders from multiplicity codes.
\newblock In {\em Approximation, Randomization, and Combinatorial Optimization. Algorithms and Techniques (APPROX/RANDOM 2022)}. Schloss Dagstuhl--Leibniz-Zentrum f{\"u}r Informatik, 2022.

\bibitem[PR01]{panconesi2001some}
Alessandro Panconesi and Romeo Rizzi.
\newblock Some simple distributed algorithms for sparse networks.
\newblock {\em Distributed computing}, 14(2):97--100, 2001.

\bibitem[SB24]{behnezhad2023streaming}
Mohammad Saneian and Soheil Behnezhad.
\newblock Streaming edge coloring with asymptotically optimal colors.
\newblock In {\em 51st International Colloquium on Automata, Languages, and Programming, {ICALP} 2024, July 8-12, 2024, Tallinn, Estonia}, volume 297 of {\em LIPIcs}, pages 121:1--121:20. Schloss Dagstuhl - Leibniz-Zentrum f{\"{u}}r Informatik, 2024.

\bibitem[Sha49]{shannon1949theorem}
Claude~E Shannon.
\newblock A theorem on coloring the lines of a network.
\newblock {\em Journal of Mathematics and Physics}, 28(1-4):148--152, 1949.

\bibitem[Sin19]{sinnamon2019fast}
Corwin Sinnamon.
\newblock Fast and simple edge-coloring algorithms.
\newblock {\em arXiv preprint arXiv:1907.03201}, 2019.

\bibitem[SW21]{SaberiW21}
Amin Saberi and David Wajc.
\newblock The greedy algorithm is not optimal for on-line edge coloring.
\newblock In {\em 48th International Colloquium on Automata, Languages, and Programming (ICALP)}, volume 198 of {\em LIPIcs}, pages 109:1--109:18, 2021.

\bibitem[TSUZ01]{ta2001loss}
Amnon Ta-Shma, Christopher Umans, and David Zuckerman.
\newblock Loss-less condensers, unbalanced expanders, and extractors.
\newblock In {\em Proceedings of the thirty-third annual ACM symposium on Theory of computing}, pages 143--152, 2001.

\bibitem[Viz65]{vizing1965chromatic}
Vadim~G Vizing.
\newblock The chromatic class of a multigraph.
\newblock {\em Cybernetics}, 1(3):32--41, 1965.

\end{thebibliography}


\end{document}